\documentclass[11pt]{article}

\input{setup.sty}
\def\Hlarge{\cH_{\textit{large}}}
\def\trunc{{\rm trunc}}
\def\ln{\log}
\def\IR{\text{IR}}

\title{Modular flow and the area operator}

\author[a]{Xi Dong,}
\author[a]{Donald Marolf,}
\author[b]{and Pratik Rath}

\affiliation[a]{Department of Physics, University of California, Santa Barbara, CA 93106, USA}
\affiliation[b]{Department of Theoretical Physics, Tata Institute of Fundamental Research, Homi Bhabha Rd,
Mumbai 400005, India}

\emailAdd{xidong@ucsb.edu}
\emailAdd{marolf@ucsb.edu}
\emailAdd{pratik.rath@tifr.theory.res.in}

\abstract{In the context of the AdS/CFT correspondence, the Jafferis-Lewkowycz-Maldacena-Suh (JLMS) formula relates the boundary modular Hamiltonian to the HRT area operator and the corresponding bulk modular Hamiltonian. In standard discussions of a holographic code built from perturbative excitations on a fixed classical background, the $A/4G$ term dominates over the bulk modular Hamiltonian in the limit $G \to 0$. However, this setup is insufficient to describe the action of boundary modular flows that change the classical background. In contrast, our recent discussion of `large' codes provides a useful framework to obtain a JLMS formula that relates bulk and boundary modular flows acting on appropriate states in the code. Here, one sees that the bulk modular Hamiltonian is {\it not} in fact subdominant when the modular flow defined by one state acts on another state peaked around a sufficiently different classical background. Instead, the bulk modular Hamiltonian can be of order $1/G$ and can dominate over the $A/4G$ term. Nevertheless, we show that if the two states have the same classical background except for data conjugate to the HRT area then, under mild assumptions and over a large range of flow parameters, the modular-flowed state is described by a classical background given by the HRT-area flow alone.
}

\begin{document}

% Title Page
% \input{editionlegend.tex}
\maketitle

\section{Introduction}

% (fold)
\label{sec:intro}

The Jafferis-Lewkowycz-Maldacena-Suh (JLMS) \cite{Jafferis:2015del} relation forms a linchpin connecting many facets  of bulk reconstruction and quantum error correction in AdS/CFT.  This relation is valid at the level that includes the leading bulk quantum corrections, and for states in a `code subspace' of the CFT associated with appropriately semiclassical physics in the bulk.   Under such conditions, and for Einstein-Hilbert gravity with minimally-coupled matter, it reads
\begin{equation}\label{eq:JLMS}
P_{\text{code}} K_{R} P_{\text{code}}=\frac{\hat A(\gamma_R)}{4G}+ K_{R,\text{bulk}},
\end{equation}
where $R$ is a region in the CFT, $\hat A(\gamma_R)$ is the bulk operator defined by the area of the associated Hubeny-Rangamani-Takayanagi (HRT) surface $\gamma_R$ \cite{Ryu:2006bv,Ryu:2006ef,Hubeny:2007xt}, $P_{\text{code}}$ is the projection onto the chosen code subspace, and $K_R$,
$K_{R,\text{bulk}}$ are respectively the CFT (or `boundary') modular Hamiltonian in $R$ and the corresponding bulk modular Hamiltonian.\footnote{Here and below, factors of the identity operator on complementary subsystems are left implicit when appropriate. For example, the left-hand side of \eqref{eq:JLMS} could have been written more explicitly as $P_{\text{code}} \left( K_R \otimes \mathds{1}_{\bar R}\right) P_{\text{code}}$ where $\mathds{1}_{\bar R}$ is the identity on the CFT region complementary to $R$, and correspondingly for the right-hand side.}  We also remind the reader that a modular Hamiltonian $K_R$ is defined in terms of the corresponding density operator (or density matrix) $\tilde{\rho}_R$ on $R$ by\footnote{We use a tilde to indicate CFT (or `boundary') states, density operators, and Hilbert spaces, though not for the modular Hamiltonian.}
\begin{equation}
\label{eq:modularK}
K_R = -\log \tilde{\rho}_R.
\end{equation}
 For higher-derivative theories and/or non-minimal couplings, the term $\hat A/4G$ is replaced by the appropriate geometric entropy that appears in the corresponding generalization of the HRT formula. See Ref.~\cite{Dong:2013qoa,Dong:2017xht} and references thereto for studies of such corrections.

The result \eqref{eq:JLMS} follows directly \cite{Jafferis:2015del,Dong:2016eik} from the Faulkner-Lewkowycz-Maldacena (FLM) form \cite{Faulkner:2013ana} of the quantum-corrected HRT formulae
(see also~\cite{Engelhardt:2014gca})  and implies that operators in the bulk entanglement wedge can be reconstructed in $R$ via an operator-algebra quantum error-correcting code (OAQEC) \cite{Harlow:2016vwg}.
Additional bulk computations then show that the OAQEC encoding of fixed-area states is associated with an approximately flat entanglement spectrum, from which it follows, at the order considered here, that $K_R$ commutes with $P_{\text{code}}$ \cite{Dong:2018seb,Akers:2018fow,Dong:2019piw}.  As a result, when acting on a state in the code subspace, one may drop the $P_{\text{code}}$ projectors from \eqref{eq:JLMS} to write
\begin{equation}\label{eq:UPJLMS}
    K_{R} =\frac{\hat A(\gamma_R)}{4G}+ K_{R,\text{bulk}}.
\end{equation}
This unprojected form then allows one to exponentiate the JLMS formula and to thereby relate bulk and boundary modular flows. Modular flow in turn is very useful in providing explicit formulae for bulk reconstruction either directly \cite{Faulkner:2017vdd} or via the Petz map \cite{Cotler:2017erl,Chen:2019gbt}.

Taking the limit of a strictly semiclassical bulk often simplifies discussions of AdS/CFT by reducing the bulk calculation to one that is purely classical. We are thus motivated to better understand the JLMS relation \eqref{eq:JLMS} and the flow generated by $K_R$ in this limit.  Naively, one might expect that such a limit would allow one to ignore the $K_{R,\text{bulk}}$ term and thus to approximate $K_R$ by $\hat A/4G$ alone. Such a statement is often true for expectation values in the state from which $K_R$ was obtained (e.g., in the HRT formula the bulk entropy contribution is often subleading to the area contribution). 

However, we point out several ways in which one needs to be more careful with the operator version of the statement.  One such way is of course that the bulk modular Hamiltonian can still cause order-one changes in the state for order-one flow parameters $s$.  In addition, a
more striking issue arises when the modular flow acts on a boundary state $\td\t$ 
corresponding to a bulk classical background (which we denote by $\Phi_\t$) for which the entanglement wedge is distinct from that of the classical background $\Phi_\rho$ associated with the state $\tilde{\rho}$ defining $K_R$.
In that case, the bulk modular Hamiltonian is also generally of order $1/G$ and can dominate over the $\hat A/4G$ term when the two classical backgrounds are sufficiently different.

As we will see, the above statements are natural consequences of describing modular flows of appropriate semiclassical states using an unprojected JLMS formula of the form \eqref{eq:UPJLMS}. However, there are two issues that arise in attempting to make such statements precise. The first is that standard derivations of the exponentiated JLMS relation apply only in code subspaces built from perturbations around a given classical background \cite{Almheiri:2014lwa,Dong:2016eik}. In particular, this forbids any application to the above modular flows where the classical backgrounds $\Phi_\rho$ and $\Phi_\t$ are distinct. In addition, even for $\Phi_\rho=\Phi_\t$, one should recall that the Hamiltonian flow of the HRT area observable on classical phase space modifies the classical background \cite{Bousso:2020yxi,Kaplan:2022orm,Dong:2025orj}. As a result, the quantum flow generated by the right-hand side of \eqref{eq:UPJLMS} cannot preserve any such code subspace. While this is not necessarily a problem for infinitesimal flows, the usual JLMS setting simply does not allow discussion of such flows at order-one flow parameter $s$.

The second issue is that, again even for flows in which $s$ is only of order $G^0$, the fact that the above semiclassical setting necessarily involves exponentially small probabilities raises another technical problem for the usual derivation of either JLMS or its exponentiated form. The point is that the functions $\ln \rho$ and $\rho^{is}$ (for real $s$) are both singular at $\rho=0$, and thus are very sensitive to small errors when the eigenvalues of $\rho$ are small. In particular, when the eigenvalues of $\rho$ are already exponentially small due to semiclassical effects, there is a potential for other effects that are non-perturbatively small to have significant impact. This feature lies at the heart of the failure of JLMS described in \cite{Kudler-Flam:2022jwd}, and it can clearly also propagate into discussions of modular flow.

It turns out that one can deal with both of these issues using the so-called `large code' framework recently formulated in Ref.~\cite{JLMS}.  This framework was explicitly constructed to include an appropriate class of states describing a range of distinct backgrounds, and to thus provide a suitably generalized understanding of the bulk modular Hamiltonian. Ref.~\cite{JLMS} also established controlled versions of the exponentiated JLMS relation in this framework. We review the needed framework and assumptions in section~\ref{sec:JLMS} and apply the results in section~\ref{sec:JLMSflow}. We warn the reader that this material is somewhat technical; readers happy to assume the various technicalities may wish to skip ahead to the more physical discussion in section~\ref{sec:largeKbulk}.

In particular, using the framework of our large code it is easy to see that the naive expectation that $K_R$ be well approximated by $\hat A/4G$ alone generally fails to be correct. Such a statement would imply that all semiclassical states can be approximated by a single universal density operator \begin{equation}
\label{eq:naive}
    \tilde{\rho}^{naive}_R = {\cal N} e^{-\frac{\hat A}{4G}},
\end{equation} for an appropriate normalization constant ${\cal N}$.  But this is clearly false, as semiclassical states can generally be distinguished by their semiclassical one-point functions, and also by their boundary von Neumann or R\'enyi entropies.\footnote{As we review in section~\ref{sec:largeKbulk}, a concrete counter-example to \eqref{eq:naive} is that, for BTZ black holes and neglecting subleading contributions from boundary gravitons, the modular Hamiltonian is in fact quadratic in $\hat A$; see \er{eq:KRBTZ} and e.g.\ Ref.~\cite{Marolf:2020vsi}.}

We will argue in section~\ref{sec:largeKbulk} below that the discrepancy between \eqref{eq:naive} and the actual CFT state $\tilde{\rho}_R$ is nevertheless in accord with \eqref{eq:UPJLMS}.  In particular, the difference is given by the notion of $K_{R,\text{bulk}}$ defined on the large code, which can have eigenvalues of size $O(1/G)$ coming from the tail of the bulk entanglement spectrum.

However, the fact that $K_{R,\text{bulk}}$ can be important even in the semiclassical limit then raises another puzzle.  As we will review in section~\ref{sec:KandA}, Ref.~\cite{Bousso:2020yxi} provided highly non-trivial evidence that the flow generated by $K_R$ on the original boundary state $\td\r$ is semiclassically associated with a so-called boundary-condition-preserving (BCP) kink transformation.  This transformation is indeed the classical flow generated by $\hat A/4G$ \cite{Kaplan:2022orm,Dong:2025orj}, leaving no room for a nontrivial action of $K_{R,\text{bulk}}$.  We resolve this puzzle in section~\ref{sec:KandA} by showing that, although the operator $K_{R,\text{bulk}}$ remains non-trivial in the bulk semiclassical limit, if we are given two
states $\rho,\t$ that describe the same restricted classical background (in the sense defined in section~\ref{sec:overview}), then under appropriate conditions and for flow parameters $s=O(G^0)$, the modular-flowed states
\begin{equation}
\td\t_s :=    \tilde{\rho}_R^{is}\,\td\t\, \tilde{\rho}_R^{-is} = e^{-isK_R}\td\t e^{isK_R}
\end{equation}
do indeed describe classical backgrounds $\Phi_{\t,s}$ obtained from $\Phi_\t$ by the relevant BCP kink-transformation.

We close with a summary and brief remarks in section~\ref{sec:discussion}.
We comment on applications of our results to states that are superpositions of classical spacetimes as well as situations with non-trivial quantum extremal surfaces. Note that, throughout this work, we consider only leading-order bulk quantum corrections.  All higher-order such corrections will thus be neglected without further comments.

% section introduction (end)
\section{Large codes with JLMS}
\label{sec:JLMS}

The traditional discussion of the holographic duality as a quantum error-correcting code is based on picking a particular classical background and working with a small code built from nearby quantum excitations \cite{Almheiri:2014lwa,Dong:2016eik}. However, as discussed above, this is insufficient to understand the holographic dual of boundary modular flow, since the action of the area operator changes the classical background. For this purpose, a large code was constructed in Ref.~\cite{JLMS}. 

Furthermore, under appropriate conditions, this large code was shown to satisfy an approximate version of the exponentiated JLMS relation. This means that the boundary modular flow operator $\tilde{\rho}_R^{is}$ is well-approximated on this code by the `JLMS flow' operator $e^{-is\hat A_{cg}/4G}\rho_r^{is}$, where $\hat A_{cg}$ is a coarse-grained area operator defined on the large code (hence the subscript `cg') and $\rho_r^{is}=e^{-isK_{\rho_r}}$ is the corresponding bulk modular flow operator.\footnote{From this section onward, we use $r$ to refer to bulk quantities defined by the region $R$.  One may say that $r$ refers to the corresponding bulk entanglement wedge, though our bulk states may describe superpositions and/or mixtures of distinct classical backgrounds and thus the term `entanglement wedge' need not define a single domain of dependence.}  Here $K_{\rho_r}:=-\log\rho_r$ is the bulk modular Hamiltonian previously denoted by $K_{R,\text{bulk}}$. We will use the term `JLMS flow' to refer to the flow generated by $\hat A_{cg}/4G + K_{\rho_r}$.

We first review the large-code framework and then describe the exponentiated JLMS results. We refer the reader to \cite{JLMS} for further relevant details.

\subsection{Large-code framework}
\label{sec:overview}

We work throughout with UV cutoffs in both the CFT and the bulk.  In particular, we treat the boundary Hilbert space $\tH$ as a tensor product $\tH_R\otimes\tH_{\bar R}$.  As usual, we ignore potential issues associated with gauge invariance in the CFT that might otherwise obstruct such a factorization; such issues may be dealt with by replacing the tensor product with a sum of tensor products.

In the AdS/CFT context, the construction of \cite{JLMS} began by considering states that, at small $G$, describe small fluctuations about a given set of classical backgrounds\footnote{Large fluctuations of observables localized away from the HRT surface can also be accommodated under appropriate conditions \cite{JLMS}.} and `chopping' such states into pieces associated with some given set of windows $[A-\epsilon/2,A+\epsilon/2)$ of eigenvalues for the HRT area operator of interest. The various states associated with a common window were then collected together into what we may call a `small' code. In this sense, it followed the recipe for small codes given in \cite{Dong:2019piw}.  
Such chopping into windows is needed because the unprojected JLMS relation \eqref{eq:UPJLMS} holds to good approximation within each small code only when the associated entanglement spectrum is approximately flat, leading to an approximately $n$-independent code contribution to the boundary R\'enyi entropy. As we will see below, this chopping will also be the key ingredient that allows us to discuss a sufficiently large set of classical backgrounds to accommodate the desired modular flows.

Ref.~\cite{JLMS} then further decomposed each area-window sector so that every resulting small code Hilbert space\footnote{As in \cite{JLMS}, a code Hilbert space means what is often called the `logical' Hilbert space for a code, and a code subspace means the corresponding subspace of the `physical' Hilbert space (which is the CFT Hilbert space here).} $\cH^\a$ had an exact tensor-product structure
\begin{equation}
\label{eq:factoralpha}
\cH^\a = \cH_r^\a \otimes \cH_{\bar r}^\a
\end{equation}
associated with the bulk entanglement wedges $r, \bar r$ defined by the complementary pair of boundary regions $R, \bar R$. Here the index $\a$ labels the various small code Hilbert spaces $\cH^\a$, which we will also refer to as $\a$-sectors. Each small code is associated with an area-window that we call
\begin{equation}\la{eq:areaWindowalpha}
[A_\a-\e_\a/2,A_\a+\e_\a/2).
\end{equation}
Therefore, on each given $\cH^\a$ the geometric area operator $\hat A$ is approximated by the c-number $A_\a$ when the window width $\e_\a$ is small.
We allow the small code Hilbert spaces $\cH^\a$ to be possibly infinite-dimensional.

By a version of the standard path integral argument of \cite{Faulkner:2013ana}, an FLM relation of the following form then holds to good accuracy on each small code $\cH^\a$:
\begin{equation}
\label{eq:FLM}
S_{R} = \frac{A_{cg}^\a}{4G} + S_{\text{bulk}}.
\end{equation}
Here $A_{cg}^\a$ is a coarse-grained area for this small code, defined as\footnote{This $A_{cg}^\a$ was simply called $A^\a$ in \cite{JLMS}, which was taken there to include the $\log \e_\a$ correction.}
\begin{equation}
\label{eq:Aalphacg}
A_{cg}^\a:=A_\a+4G\log\e_\a,
\end{equation}
where the area-window width $\e_\a$ inside the logarithm is made dimensionless using an implicit reference scale. The logarithmic correction here is due to the area window \er{eq:areaWindowalpha} and can be understood as arising from the normalization of the fixed-area path integral: in deriving \er{eq:FLM} using the replica method \cite{Lewkowycz:2013nqa}, the $n$-replica path integral $Z_n$ gets a factor of $\e_\a$ when integrating the HRT area over the window, and the boundary entropy $S_R = -\lim_{n\to 1} \pa_n \log (Z_n/Z_1^n)$ receives an additive correction of $\log\e_\a$.
If the small code $\cH^\a$ is specified by constraining other bulk variables in addition to the area window, the definition of $A_{cg}^\a$ would receive additional logarithmic contributions.
An approximate projected JLMS relation then follows under appropriate conditions \cite{JLMS}.

However, 
as described in \cite{JLMS}, in order to exponentiate the JLMS relation it is important that each window-width $\epsilon$ vanish {\it faster} than $G$ in the limit $G\rightarrow 0$; i.e., we take $\epsilon = o(G)$. This has two important effects.  The first is simply that there will be many windows associated with each classical background.

To describe the second effect, we define the term `field data' to mean
the classical metric and matter fields together with the appropriate derivatives needed to specify a classical solution.  We also normalize the light local bulk fields, including the metric, so that the entire low-energy bulk action is $1/G$ times an action that remains finite as $G\rightarrow 0$.  With this normalization, the Poisson brackets among the corresponding field data are of order $G$.  The second effect is then that any field data whose Poisson bracket with $A$ is nonzero at order $G$ (including the `conjugate' of $A$) will have large quantum uncertainties in a state supported in a single window.  As a result, even in the limit $G\rightarrow 0$, states in a single window describe the original classical background only for field data whose Poisson brackets with $A$ vanish.  These include all field data in the interiors of either bulk entanglement wedge defined by $\gamma_R$ and, on $\gamma_R$ itself, functions of the induced metric, including $A$ itself.  We will therefore use the symbol $\bar \Phi$ (with an over-bar, and perhaps with additional decorations) to denote the $G\to 0$ value of such a restricted set of field data that Poisson-commute with $A$ (assuming their fluctuations vanish in this limit), and call it a `restricted classical background.'  We will use $\Phi$ to denote the $G\to 0$ value of the complete field data (which in particular include those data to the past and future of $\gamma_R$), and call it a `complete classical background' (again assuming the field data have vanishing fluctuations in this limit). We will construct our small code Hilbert spaces such that all states in a given $\cH^\a$ are associated with a restricted classical background which is completely determined by $\a$ and which we denote by $\bar \Phi_\a$.\footnote{It is possible to construct small code Hilbert spaces such that a given $\cH^\a$ contains multiple backgrounds, although we will not pursue this possibility here.}

For a given collection of restricted classical backgrounds $\bar \Phi_\a$ and area-windows \er{eq:areaWindowalpha}, Ref.~\cite{JLMS} then defined a large code Hilbert space $\Hlarge$ by simply taking the orthogonal direct sum of the associated small code Hilbert spaces:
\be\la{eq:Hlargedef}
\Hlarge := \bigoplus_\a \cH^\a.
\ee
The coarse-grained area operator $\hat A_{cg}$ on this large code is then defined as
\begin{equation}\la{eq:Acgdef}
\hat A_{cg}:=\bigoplus_\a A_{cg}^\a\mathds{1}^\a,
\end{equation}
with $A_{cg}^\a$ defined by \er{eq:Aalphacg}.
The large code comes with a linear encoding map $V: \Hlarge \to \tH$ into the boundary Hilbert space, which we will discuss in more detail in section~\ref{sec:JLMSflow}.

By construction, the space $\Hlarge$ now includes potentially different restricted classical backgrounds $\bar \Phi_\a$.  Moreover, so long as we include a sufficiently-complete set of area-windows, states describing a complete classical background $\Phi$ can then be constructed by superposing states in many different small code Hilbert spaces $\cH^\a$.

In particular, by allowing relative phases between the $\cH^\a$, we now also have access to an approximate `conjugate' variable to $A$ that we may call $\eta$.  It is useful to recall from \cite{Kaplan:2022orm} that (as anticipated in
\cite{Jafferis:2014lza,Ceyhan:2018zfg,Faulkner:2018faa,Bousso:2019dxk,Bousso:2020yxi,Lewkowycz:2018sgn,Chen:2018rgz}) at the classical level this $\eta$ is naturally taken to measure the relative boost angle between the two wedges; see figure~\ref{fig:areaflow}.  As a result it can also be described as quantifying a certain generalized notion of `time-shift' between boundary region $R$ and the complementary boundary region $\bar R$ (i.e., it generalizes the operators described in \cite{Thiemann:1992jj,Kastrup:1993br,Kuchar:1994zk,Harlow:2018tqv}), though it multiplies this time-shift by an analogue of the surface-gravity that would be defined for a bifurcate event horizon.\footnote{Note also that the operator $\hat \eta$ or equivalently the time shift operator $\hat \delta$ is not self-adjoint and, thus, the decomposition into its eigenstates is not precise \cite{Harlow:2018tqv}. In the classical description, the issue can be seen from the fact that there is a phase space boundary at $A=0$ which the flow generated by $\eta$ can reach at finite flow parameter. However, as long as we consider wavefunctions supported well away from $A=0$, we may nevertheless act on the wavefunction with $\hat \eta = 4G i \partial/\partial A$; any effects associated with the lack of a self-adjoint extension will be small.  Moreover, in our case, since we have a coarse-grained area operator, this $\hat \eta$ is also coarse-grained.}

\begin{figure}[t]
       \centering
        \includegraphics[width=0.6
        \textwidth]{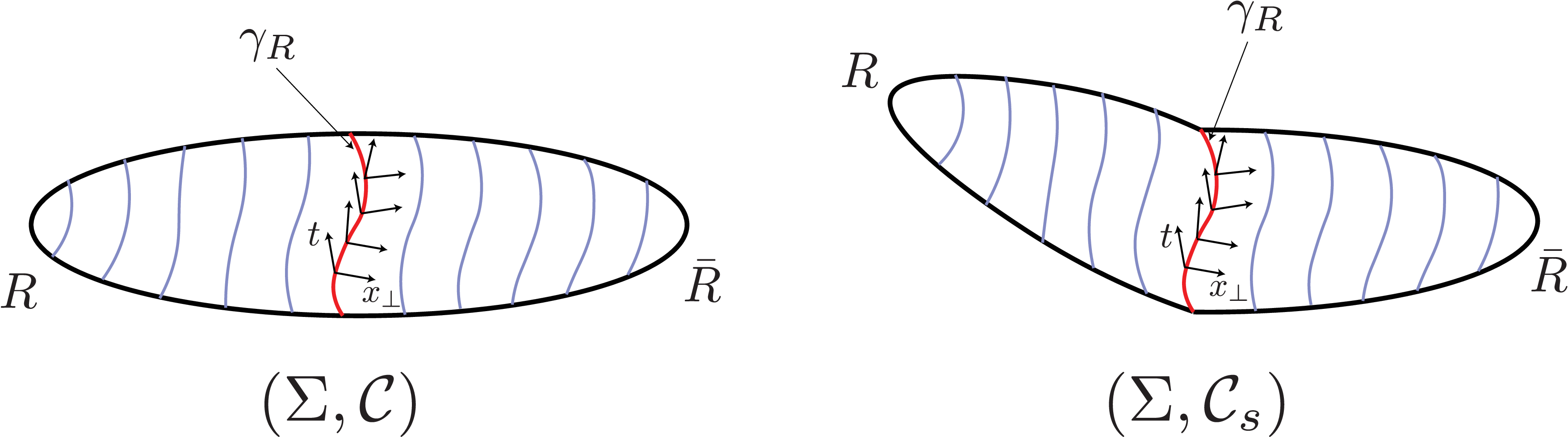}
        \caption{The classical Hamiltonian flow generated by $A/4G$ modifies the initial data $\mathcal{C}$ on a Cauchy slice $\Sigma$ to new data $\mathcal{C}_s$ by introducing a delta function in the extrinsic curvature component $K_{\perp \perp}$ where $x_{\perp}$ represents the intrinsic coordinate on $\Sigma$ that is orthogonal to the HRT surface $\gamma_R$. The Cauchy data in each separate entanglement wedge $r, \bar r$ remain unchanged.  In particular, the slice $\Sigma$ with the new data $\mathcal{C}_s$ remains glued to the boundary at a boundary location that does not depend on $s$. For this reason, the action of this flow is known as a boundary-condition-preserving kink transformation \cite{Kaplan:2022orm}.}
        \label{fig:areaflow}
\end{figure}

\subsection{Exponentiated JLMS}
\label{sec:expJLMS}

The above large code is a setting where one can discuss both modular flow and the coarse-grained area flow on $\Hlarge$.  In particular, $\hat A_{cg}$ is a c-number on each small code $\cH^\a$. 
The associated area flow operator $e^{-is\hat A_{cg}/4G}$ on $\Hlarge$ thus acts only
by changing the relative phase between small codes with different $A_{cg}^\a$.

Under suitable conditions, Ref.~\cite{JLMS} established two forms of the exponentiated JLMS relation on the large code: an expectation-value bound (Theorem~4) and a stronger operator-norm bound (Theorem~6). To describe the relevant conditions, we first use the decomposition \er{eq:Hlargedef} to write a state $\rho$ on $\Hlarge$ in the associated block form with diagonal blocks $p_\a\rho^\a$ and possible off-diagonal blocks:
\begin{equation}\la{rdecom}
\rho =
\begin{pmatrix}
\rho^{\a\a} & \rho^{\a\b} & \cd\\
\rho^{\b\a} & \rho^{\b\b} & \cd\\
\cd & \cd & \cd
\end{pmatrix}
,\qquad
\rho^{\a\a}=p_\a\rho^\a,\qquad  p_\a := \tr \r^{\a\a},
\end{equation}
where the probability (or weight) $p_\a$ satisfies $\sum_\a p_\a = 1$, and $\r^\a$ is a normalized state on $\cH^\a$ proportional to $\r^{\a\a}$ when it is nonzero.
Tracing over $\bar r$ by definition removes the off-diagonal blocks and gives
\begin{equation}
\label{eq:rhor}
\rho_r=\bigoplus_\a p_\a\rho_r^\a.
\end{equation}
Here we have used \eqref{eq:factoralpha} to define the normalized density operator $\rho_r^\a$ on $\cH_r^\a$.  It follows that
\begin{align}
\rho_r^{is}
&=\bigoplus_{\a}
\left(p_\a\rho_{r}^\a\right)^{is},\\
\label{eq:Kbulkblocks}
K_{\rho_r}:=-\log\rho_r
&=\bigoplus_{\a}
\left[(-\log p_\a)\mathds{1}_r^\a
+K_{\rho_r^\a}\right],
\end{align}
where $\mathds{1}_r^\a$ is the identity on $\cH_r^\a$, and $K_{\rho_r^\a}:=-\log\rho_r^\a$.
Note also that we may use
the encoding map $V$ introduced above to define a boundary state $\td\r$ corresponding to $\r$ as the properly normalized density operator proportional to $V\r V^\dag$. In the following, we use a particular state $\r$ and its corresponding boundary state $\td\r$ to define the bulk and boundary modular Hamiltonians in the JLMS relation.

In a setting of this form, Theorem~4 of Ref.~\cite{JLMS} establishes exponentiated JLMS at the level of expectation values in certain test states $\s$ that behave well with respect to $\r$.\footnote{Throughout this paper, we reserve $\s$ for such test states and use $\t$ for the state on which modular flow acts.} To understand this, we first decompose the test state $\s$ on $\Hlarge$ as in \er{rdecom}, with
\begin{equation}\la{sdecom}
\s^{\a\a}=w_\a\s^\a,\qquad
w_\a := \tr \s^{\a\a},
\end{equation}
where the probability $w_\a$ will also be called the $\s$-weight of that sector. One important condition in Theorem~4 is a notion of relative log-stability in the small codes defined as follows. For a small code labelled by $\a$, we say that $\rho^\a$ is relatively log-stable with respect to $\sigma^\a$ if, for some $\epsilon_{tail}>0$,
\begin{align}
\label{eq:relativelogstabilitybulk}
\<(\rho_r^\a)^{-1}\>_{(\sigma_r^\a)^2}
&\leq\frac{1}{\epsilon_{tail}},\\
\label{eq:relativelogstabilityboundary}
\<(\tilde\rho_R^\a)^{-1}\>_{(\tilde\sigma_R^\a)^2}
&\leq\frac{1}{\epsilon_{tail}},
\end{align}
where angle brackets denote expectation values (as defined in Ref.~\cite{JLMS}), and $\tilde\rho_R^\a, \tilde\sigma_R^\a$ are the boundary reduced density operators corresponding to $\r^\a, \s^\a$, respectively.
These two inequalities control how strongly $\sigma_r^\a$ probes the small-eigenvalue tail of $\rho_r^\a$, and similarly on the boundary.
In particular, \eqref{eq:relativelogstabilitybulk} and \eqref{eq:relativelogstabilityboundary} respectively require the supports of $\sigma_r^\a$ and $\tilde\sigma_R^\a$ to be contained in the supports of $\rho_r^\a$ and $\tilde\rho_R^\a$.
If this relative log-stability condition is satisfied in all $\a$-sectors (except for a set carrying small total $\s$-weight), and assuming other appropriate conditions (as will be stated in section~\re{subsec:general} when we use the theorem), Theorem~4 gives an exponentiated JLMS relation in the form of a bound on the expectation value in $\s$ of the following operator:
\begin{equation}\la{JLMSdiffop}
V^\dag \tilde{\rho}_R^{is} V - e^{-is\frac{\hat A_{cg}}{4G}}\rho_r^{is},
\end{equation}
where the first term is the boundary modular flow operator pulled back to the large code Hilbert space by the encoding map $V$. We will use this theorem for the version of our discussion in section~\re{subsec:general}. In particular, it does not require the Hilbert spaces $\cH_r^\a$ to be finite-dimensional.

Theorem~6 of Ref.~\cite{JLMS} gives a stronger operator-norm version of exponentiated JLMS.  It bounds the operator norm of \er{JLMSdiffop}; equivalently, it bounds the expectation value of \er{JLMSdiffop} in all states $\s$ uniformly.  This version requires an (absolute) notion of log-stability in each small code: we say that $\rho^\a$ is log-stable if it is relatively log-stable with respect to every state $\sigma^\a$ on $\cH^\a$, with the same $\epsilon_{tail}$. In particular, satisfying the bulk side \er{eq:relativelogstabilitybulk} of the condition for every state $\sigma^\a$ is equivalent to the statement that all eigenvalues of $\rho_r^\a$ are at least $\epsilon_{tail}$. This necessarily requires $\cH_r^\a$ for each $\a$ to be finite-dimensional. We will use this theorem for the version of our discussion in section~\re{subsec:truncated}.

\section{Approximating boundary modular flow by the JLMS flow}
\label{sec:JLMSflow}

We wish to consider an AdS/CFT context in which we are given two CFT states and a region $R$ in the CFT.  We call the state whose reduced density operator defines the modular flow the reference state $\td\r$, and the state on which modular flow acts the target state $\td\t$. For the moment, let us focus on the case where both states are pure, writing $\td\r=|\tilde \psi \rangle \langle\tilde \psi |$ and $\td\t=|\tilde \phi\rangle\langle\tilde \phi|$. 
We then wish to understand the boundary modular flow $\tilde\rho_R^{is}|\tilde\phi\rangle=e^{-isK_{R,\psi}}|\tilde\phi\rangle$, where $K_{R,\psi}:=-\log\tilde\rho_R$.
Since this boundary modular flow depends only on $\tilde\rho_R$, we may choose to work with a convenient $|\tilde\psi\rangle$ among the purifications of $\tilde\rho_R$.

We assume that the states $|\tilde \psi \rangle, |\tilde \phi\rangle$ admit path integral constructions that can be translated into bulk language using sources only for low energy fields. We use $|\psi\rangle, |\phi \rangle$ as shorthand for the resulting semiclassical bulk path-integral descriptions until we construct their code representatives below, and we use these path integrals through their perturbative small-$G$ expansions.
In particular, the path integrals for $|\psi \rangle, |\phi\rangle$ do {\it not} initially include constraints restricting the HRT area to parametrically-small windows (or similar constraints on other field variables), though below we will impose such constraints to decompose them into area-window states.  In this sense, $|\psi \rangle, |\phi\rangle$ may be called `fully semiclassical,' as opposed to being semiclassical only in the union of the entanglement wedges $r$ and $\bar r$.

Anticipating the construction below, we write $\rho=|\psi\rangle\langle\psi|$ for the reference state. Our main statement of this section will be that for $s=O(G^0)$, we have\footnote{Since the logarithm of zero is ill-defined, the literal modular and JLMS flow operators are also ill-defined on subspaces where the underlying density operator has no support.  
As in Ref.~\cite{JLMS}, in such cases we suppose that the definition of each such operator has been extended to these subspaces in such a way that the full operator is unitary. Our assumptions below will ensure that the choice of extension does not affect the approximate relation \eqref{eq:modflowisJLMS}. For simplicity, we extend the bulk modular flow to be block diagonal in $\a$ so that it commutes with $\hat A_{cg}$.\label{foot:extension}}
\begin{equation}
\label{eq:modflowisJLMS}
e^{-isK_{R,\psi}}|\tilde\phi\rangle
\ap
V e^{-is\left(\frac{\hat A_{cg}}{4G}+K_{\rho_r}\right)}|\phi\rangle,
\end{equation}
where $\hat A_{cg}$ is the coarse-grained area operator \er{eq:Acgdef}.  We will construct a large code, its encoding map $V$, and the bulk reduced state $\rho_r$ in order to establish \eqref{eq:modflowisJLMS}.

The first version of our argument is given in section~\ref{subsec:general} below.  This version uses the bulk path integral to construct a natural large code
adapted to $|\td\psi\rangle$ and applies exponentiated JLMS at the level of expectation values to establish \eqref{eq:modflowisJLMS} for target states $|\td\phi\rangle$, under an appropriate `alignment' condition with respect to $|\td\psi\rangle$ as specified below.  Section~\ref{subsec:truncated} then gives an optional procedure for truncating each small code Hilbert space to a finite-dimensional space so that the operator-norm version of exponentiated JLMS can instead be applied to establish the truncated analogue of \eqref{eq:modflowisJLMS}. While the truncation introduces a degree of arbitrariness into the construction, it then provides a sense in which \eqref{eq:modflowisJLMS} holds for all choices of target state in the truncated Hilbert space, without requiring any notion of `alignment' with $|\td\psi\rangle$.

Although we focus below on pure reference and target states, our results also extend to mixed states under similar assumptions, as discussed below.

\subsection{An untruncated construction}
\label{subsec:general}

We construct an appropriate code adapted to the bulk path integral associated with $|\td\psi\rangle$ as follows.

First choose a complete set of non-overlapping area windows
\begin{equation}
\label{eq:areaWindows}
[A_\a-\e_\a/2,A_\a+\e_\a/2),
\end{equation}
labeled by some index $\a$. For each $\a$, consider a constrained bulk path integral that has the same sources as the original path integral defining $|\psi\rangle$, together with a constraint that the HRT area lie in the corresponding area window. The perturbative small-$G$ expansion of this constrained path integral defines an unnormalized state that we may write as $\sqrt{p_\a}|\psi^\a\rangle$, where $|\psi^\a\rangle$ is normalized and $\sum_\a p_\a=1$. The constrained saddle determines a restricted classical background $\bar\Phi_\a$, and the perturbative expansion defines $|\psi^\a\rangle$ as an element of a perturbative Hilbert space $\cH^\a$ of gravitons and matter fields on $\bar\Phi_\a$, with a fixed bulk UV cutoff. Perturbatively varying the sources gives the other states in $\cH^\a$.

Each perturbative Hilbert space has the tensor-product structure
\begin{equation}\la{eq:Hafactor}
\cH^\a=\cH_r^\a\otimes\cH_{\bar r}^\a,
\end{equation}
or a direct sum of such tensor products. If it is such a direct sum labeled by some index $\mu$, we refine the index $\a$ to include $\mu$, and \er{eq:Hafactor} holds for the refined $\a$. In that case, distinct $\a$-sectors may have the same associated area window.

For each such perturbative Hilbert space $\cH^\a$, we define a small code, with $\cH^\a$ identified with its code Hilbert space. The perturbative holographic dictionary defines an encoding map $V_\a:\cH^\a\rightarrow\tH$ into the CFT Hilbert space.

We then define the large code by taking the direct sum of this (possibly infinite) set of small codes:
\begin{equation}
\Hlarge:=\bigoplus_\a\cH^\a,
\qquad
\hat A_{cg}:=\bigoplus_\a A_{cg}^\a\mathds{1}^\a.
\end{equation}
Here the eigenvalues $A_{cg}^\a$ of the coarse-grained area operator $\hat A_{cg}$ are defined in \eqref{eq:Aalphacg}.  The large-code encoding map $V:\Hlarge\to\tH$ is defined by combining the encoding maps $V_\a$ for the small codes.  Because the set of windows is complete, the constrained path integrals together reconstruct the original bulk path integral perturbatively in $G$. We therefore define the reference state on $\Hlarge$ by
\begin{equation}
|\psi\rangle=\bigoplus_\a\sqrt{p_\a}\,|\psi^\a\rangle,
\end{equation}
whose encoding reproduces the prescribed boundary reference state $|\tilde\psi\rangle$ up to normalization:
\begin{equation}
V|\psi\rangle\propto|\tilde\psi\rangle.
\end{equation}
Obtaining the prescribed reference state $|\tilde\psi\rangle$ exactly as an encoded state is important because a small change in $|\tilde\psi\rangle$ need not give a small change in the boundary modular flow $\tilde\rho_R^{is}$ when $\tilde\rho_R$ has very small eigenvalues.  This is why the construction is adapted to $|\td\psi\rangle$, rather than treating the two states symmetrically.

Having defined the above large code, we now consider situations in which the desired boundary target state $|\td\phi\rangle$ is at least approximately an encoded state -- equivalently, there exists a state in $\Hlarge$ (which we may call $|\phi\rangle$) whose encoding $V|\phi\rangle$ is approximately $|\td\phi\rangle$ up to normalization. A concrete way of understanding this requirement is to use the bulk path integral associated with $|\tilde\phi\rangle$ together with area-window constraints to define constrained states (similar to $|\psi^\a\rangle$ above). In general, these constrained states need not belong to the corresponding $\cH^\a$, since each $\cH^\a$ was defined to consist of perturbative states on the restricted background of $|\psi^\a\rangle$. We therefore restrict to cases in which these constrained states are in the corresponding $\cH^\a$ (which we may then write as $\sqrt{q_\a} |\phi^\a\rangle$), for all $\a$ except for a set carrying small total target-state weight.\footnote{As noted above, we may choose any convenient $|\tilde\psi\rangle$ among the purifications of $\tilde\rho_R$, a choice we now use to help satisfy this requirement.}
Discarding the small-weight set of $\a$-sectors just described and renormalizing the remaining weights so that $\sum_\a q_\a=1$, we define
\begin{equation}
|\phi\rangle=\bigoplus_\a\sqrt{q_\a}|\phi^\a\rangle\in\Hlarge.   
\end{equation}
Its encoded state $V|\phi\rangle$ is approximately the desired $|\td\phi\rangle$, since we assume that the holographic map is sufficiently isometric in the current context.
This approximation causes only a small error in \eqref{eq:modflowisJLMS}, since both the boundary modular flow and the JLMS flow are unitary.

Writing $\r=|\psi\rangle\langle\psi|$, we now use Theorem~4 of Ref.~\cite{JLMS}, which gives exponentiated JLMS at the level of expectation values in certain test states denoted by $\sigma$.
For reasons that will be clear momentarily, we apply the theorem at a fixed $s$ to the following four pure test states:\footnote{For any of the four $|\chi\rangle$ that happens to vanish, our desired \er{lgexp} holds automatically without needing to use the theorem.}
\begin{equation}\la{eq:statecombine}
\s_\chi=\fr{|\chi\rangle\langle\chi|}{\langle\chi|\chi\rangle},\qqu
|\chi\rangle = |\phi\rangle\pm U_s|\phi\rangle,\,\,\,\,
|\phi\rangle\pm iU_s|\phi\rangle,
\end{equation}
where $U_s$ is the JLMS flow operator:
\begin{equation}
\label{eq:coarseJLMSflow}
U_s:=
e^{-is\left(\frac{\hat A_{cg}}{4G}+K_{\rho_r}\right)}
=e^{-is\frac{\hat A_{cg}}{4G}}\rho_r^{is}.
\end{equation}
Recall that $\r^\a$ is determined from $\r$ by \er{rdecom}, and $\s_\chi^\a$ is defined similarly.
In order to use the theorem, we assume for each of the test states $\s_\chi$ above that in every relevant small code labeled by $\a$, the encoding map $V_\a$ acts on $\rho^\a$ and $\s_\chi^\a$ as an approximate isometry, that the R\'enyi FLM relation at $n=1+is$ holds for mixtures of $\rho^\a$ and $\s_\chi^\a$, and that $\rho^\a$ is relatively log-stable with respect to $\s_\chi^\a$ in the sense of \eqref{eq:relativelogstabilitybulk} and \eqref{eq:relativelogstabilityboundary}.  We also need to assume that distinct small code subspaces are approximately orthogonal on both $R$ and $\bar R$ (in the precise exponentiated sense required by the theorem, so they remain approximately distinguishable on either boundary subsystem), and that $\s_\chi$ is aligned with $\rho$ (meaning that $\s_\chi$ assigns only a small total probability to $\a$-sectors in which $\rho$ has a small probability).

With these assumptions, the theorem gives the following exponentiated JLMS result for each choice of $|\chi\rangle$ in \er{eq:statecombine}:
\begin{equation}
\label{lgexp}
\left|
\langle V\chi|e^{-isK_{R,\psi}}|V\chi\rangle
-\langle\chi|U_s|\chi\rangle
\right|\leq \varepsilon \langle\chi|\chi\rangle \leq 4\varepsilon,
\end{equation}
where the small parameter $\varepsilon$ collects the various error terms appearing in the theorem, and in the last step we used $\langle\chi|\chi\rangle \leq 4$ due to the unitarity of $U_s$.  Importantly, this expectation-value result does not require $\cH_r^\a$ to be finite-dimensional.

To turn \eqref{lgexp} into our desired statement \er{eq:modflowisJLMS} about the flows of $|\phi\rangle$, we now apply the complex polarization identity
\begin{equation}
    \langle \alpha|\beta \rangle = \frac{1}{4}\left(\bigl\||\alpha\rangle + |\beta\rangle\bigr\|^2 - \bigl\||\alpha\rangle - |\beta\rangle\bigr\|^2 - i\bigl\||\alpha\rangle + i|\beta\rangle\bigr\|^2 + i\bigl\||\alpha\rangle - i|\beta\rangle\bigr\|^2 \right),
\end{equation}
to the difference of the two sesquilinear forms in \eqref{lgexp}, with $|\alpha\rangle=U_s|\phi\rangle$ and $|\beta\rangle=|\phi\rangle$.  This gives
\begin{equation}
\left|
\langle VU_s\phi|e^{-isK_{R,\psi}}|V\phi\rangle
-\langle U_s\phi|U_s|\phi\rangle
\right|
\lesssim\varepsilon.
\end{equation}
Since $e^{-isK_{R,\psi}}$ is unitary and $V$ acts on the states \er{eq:statecombine} as an approximate isometry, it follows that
\begin{equation}\la{eq:modflowisJLMSvec}
\bigl\|
e^{-isK_{R,\psi}}V|\phi\rangle
-VU_s|\phi\rangle
\bigr\|^2
=\bigl\|V|\phi\rangle\bigr\|^2
+\bigl\|VU_s|\phi\rangle\bigr\|^2 -2\operatorname{Re}\langle VU_s\phi|e^{-isK_{R,\psi}}|V\phi\rangle
\lesssim\varepsilon
\end{equation}
after enlarging $\varepsilon$ to include the approximate-isometry errors. Since $V|\phi\rangle\ap|\tilde\phi\rangle$, the result \er{eq:modflowisJLMSvec} establishes the desired statement \eqref{eq:modflowisJLMS}.

The result also extends to a mixed target state. Let us write
\begin{equation}
\td\t=\sum_i\lambda_i|\td\phi_i\rangle\langle\td\phi_i|,
\qquad
\t=\sum_i\lambda_i|\phi_i\rangle\langle\phi_i|,
\end{equation}
where $V|\phi_i\rangle\ap|\td\phi_i\rangle$. If the assumptions above hold for the four test states built from each $|\phi_i\rangle$ with a small weighted sum of the resulting errors, then convexity of the trace norm gives
\begin{equation}\la{eq:modflowisJLMSmixed}
\tilde\rho_R^{is}\td\t\tilde\rho_R^{-is} \ap VU_s\t U_s^\dag V^\dag,
\end{equation}
where the approximate equality holds with respect to the trace norm.

Having established the coarse-grained result \eqref{eq:modflowisJLMS} and its mixed-state extension, we now make some additional comments. We expect a fine-grained version of \eqref{eq:modflowisJLMS} to hold as well. In this fine-grained version, $\hat A_{cg}$ in \eqref{eq:modflowisJLMS} would be replaced by the geometric HRT area operator $\hat A$, and the bulk modular Hamiltonian $K_{\rho_r}$ would similarly be replaced by a fine-grained version acting on a suitable semiclassical bulk Hilbert space in which $\hat A$ may take continuous values. Thus the right-hand side of \eqref{eq:modflowisJLMS} would be replaced by a fine-grained JLMS flow followed by $V$ (interpreted as the bulk-to-boundary map).
We give some evidence in Appendix~\ref{app:sectorcoarsegraining} that this fine-grained JLMS flow is approximated by the coarse-grained version. Together with \eqref{eq:modflowisJLMS}, this provides evidence for the fine-grained version of \eqref{eq:modflowisJLMS}, although we will not rely on it in establishing the results of this paper.

For the moment, we simply note that after using \er{eq:Kbulkblocks}, the generator of the coarse-grained JLMS flow is
\begin{equation}
\frac{\hat A_{cg}}{4G}+K_{\rho_r}
=\bigoplus_\a
\left[
\left(\fr{A_{cg}^\a}{4G}-\log p_\a\right)\mathds{1}_r^\a
+K_{\rho_r^\a}
\right].
\end{equation}
The first term in square brackets is central in each $\a$-sector. Replacing this part of the generator by its fine-grained limit only causes a small error in \eqref{eq:modflowisJLMS} for $s=O(G^0)$, as long as the area-window widths are $o(G)$ (except for a set of windows carrying small total target-state weight, as they only cause a small error regardless of window width).\footnote{Here we use the ``little $o$'' notation which means that the widths should be negligible in comparison with $G$ as $G\rightarrow 0$.} To see this, we may use \er{eq:Aalphacg} to suggestively write
\begin{equation}\la{eq:Acgpa}
\fr{A_{cg}^\a}{4G} - \log p_\a = \fr{A_\a}{4G} - \log \fr{p_\a}{\e_\a}.
\end{equation}
The first term on the right-hand side is determined by the midpoint value $A_\a$ of the area window, which is close to the geometric HRT area $A$. The second term is determined by $p_\a/\e_\a$, which approaches the probability density in the fine-grained limit (when the limit exists). Thus as long as the area windows are narrow, replacing \er{eq:Acgpa} by its fine-grained limit only causes a small error in \eqref{eq:modflowisJLMS},\footnote{This is another way of understanding the $\log\e_\a$ contribution to the coarse-grained area \er{eq:Aalphacg}.} as can be seen from the standard integral formula for the difference of two unitary flows generated by Hamiltonians $H_1$ and $H_2$:
\begin{equation}\label{eq:unitaryflowdifference}
    U_1(s)-U_2(s) = -i\int_0^s U_1(s-t)\,(H_1-H_2)\,U_2(t)\,dt,\qquad
    U_i(s) := e^{-iH_i s}.
\end{equation}

\subsection{Optional truncation and operator-norm control}
\label{subsec:truncated}

We now describe an optional truncation that allows us to use the operator-norm version of exponentiated JLMS to establish a truncated analogue of \eqref{eq:modflowisJLMS}.  We then discuss when the corresponding JLMS flow is also close to the untruncated one.  Throughout, we hold fixed the boundary reference state $|\td\psi\rangle$ (and hence the boundary modular flow it defines).

For each small code labelled by $\a$, choose a finite-dimensional factorized subspace\footnote{The finite-dimensional truncation used here acts separately within each small code.  Ref.~\cite{RG} presented a different renormalization-group construction that reorganizes a large code into an IR large code.}
\begin{equation}\la{eq:Hatrunc}
\cH^\a_{\trunc}
=\cH^\a_{r,{\trunc}}\otimes\cH^\a_{\bar r,{\trunc}}
\subset\cH^\a ,
\end{equation}
and a normalized state $|\psi^\a_{\trunc}\rangle$ in this subspace that approximates $|\psi^\a\rangle$.
One may naturally ask whether such a truncated code Hilbert space allows us to describe the boundary modular flow defined by the original CFT state $|\tilde \psi\rangle$.  To argue that it does, 
we define a modification map $M_\a:\cH^\a_{\trunc}\rightarrow\cH^\a$ by
\begin{equation}
\label{eq:ModMa}
M_\a:=\iota_\a+
\left(|\psi^\a\rangle-|\psi^\a_{\trunc}\rangle\right)
\langle\psi^\a_{\trunc}|,
\end{equation}
where $\iota_\a:\cH^\a_{\trunc}\rightarrow\cH^\a$ denotes the natural inclusion.
The associated modified encoding map $V_{\a,\trunc}:\cH^\a_{\trunc}\rightarrow\tH$ is defined as
\begin{equation}
\label{eq:ModVa}
V_{\a,\trunc}:=V_\a M_\a,
\end{equation}
where $V_\a$ is the original encoding map for the untruncated small code.
By construction, $M_\a$ adds back the truncation error $|\psi^\a\rangle-|\psi^\a_{\trunc}\rangle$ so that
\begin{equation}
M_\a|\psi^\a_{\trunc}\rangle=|\psi^\a\rangle
\qqu \Rightarrow \qqu
V_{\a,\trunc} |\psi^\a_{\trunc}\rangle = V_\a|\psi^\a\rangle.
\end{equation}
Now define the truncated large code in direct analogy with the untruncated construction:
\begin{equation}
\cH_{\textit{large},{\trunc}}:=\bigoplus_\a\cH^\a_{\trunc},
\qquad
\hat A_{cg,{\trunc}}:=\bigoplus_\a A_{cg}^\a\mathds{1}_{\trunc}^\a.
\end{equation}
Its encoding map $V_{\trunc}:\cH_{\textit{large},{\trunc}}\to\tH$ is defined similarly by combining the $V_{\a,\trunc}$.
We take the reference state on this truncated large code to be $|\psi_{\trunc}\rangle$, with the same probabilities $p_\a$ as the untruncated reference state $|\psi\rangle$:
\begin{equation}
|\psi_{\trunc}\rangle
:=\bigoplus_\a\sqrt{p_\a}\,|\psi^\a_{\trunc}\rangle,
\qquad
\rho_{\trunc}:=|\psi_{\trunc}\rangle\langle\psi_{\trunc}|.
\end{equation}
It follows immediately that
\begin{equation}
V_{\trunc}|\psi_{\trunc}\rangle
=V|\psi\rangle
\end{equation}
which is precisely the original boundary state $|\tilde\psi\rangle$ after being properly normalized.  Thus the boundary modular flow it defines remains exactly unchanged.

We assume that the truncations are chosen so that there is also a normalized $|\phi_{\trunc}\rangle\in\cH_{\textit{large},{\trunc}}$ whose image under $V_{\trunc}$ approximates $|\tilde\phi\rangle$.  The truncated analogue of \eqref{eq:modflowisJLMS} that we wish to establish is
\begin{equation}
\label{eq:modflowisJLMStrunc}
e^{-isK_{R,\psi}}|\tilde\phi\rangle
\ap
V_{\trunc}e^{-is\left(\frac{\hat A_{cg,{\trunc}}}{4G}+K_{\rho_{{\trunc},r}}\right)}
|\phi_{\trunc}\rangle.
\end{equation}

To apply Theorem~6 of Ref.~\cite{JLMS}, we assume that the truncated code satisfies its hypotheses.  In particular, in each truncated small code labeled by $\a$, the theorem requires that the modified encoding map $V_{\a,\trunc}$ be an approximate isometry, that the R\'enyi FLM relation at $n=1+is$ hold for all states, and that $\rho_{\trunc}^\a$ be log-stable.  It also requires that distinct truncated small code subspaces be approximately orthogonal on both $R$ and $\bar R$, and that $\rho_{\trunc}$ be sufficiently smooth.  Here orthogonality and smoothness are understood in the precise exponentiated senses required by the theorem; see \cite{JLMS} for details.  With these assumptions, the theorem gives the following operator-norm bound:
\begin{equation}
\left\|
V_{\trunc}^\dagger e^{-isK_{R,\psi}}V_{\trunc}
-e^{-is\left(\frac{\hat A_{cg,{\trunc}}}{4G}
+K_{\rho_{{\trunc},r}}\right)}
\right\|_\infty
\leq\varepsilon,
\label{eq:opnormresult}
\end{equation}
where the small parameter $\varepsilon$ collects the various error terms appearing in the theorem.
Unlike the expectation-value result in section~\ref{subsec:general}, this operator bound can be applied to any state in the truncated large code.  Applying it to $|\phi_{\trunc}\rangle$, we find
\begin{equation}
 \left\|
V_{\trunc}^\dagger e^{-isK_{R,\psi}}V_{\trunc}|\phi_{\trunc}\rangle
-e^{-is\left(\frac{\hat A_{cg,{\trunc}}}{4G}+K_{\rho_{{\trunc},r}}\right)}
|\phi_{\trunc}\rangle
\right\|
\leq\varepsilon.
\end{equation}
Using $V_{\trunc}|\phi_{\trunc}\rangle\ap|\tilde\phi\rangle$, the approximate-isometry property of $V_{\trunc}$, and the unitarity of the boundary modular flow and the JLMS flow, we find \eqref{eq:modflowisJLMStrunc}.

This result immediately extends to a mixed target state. Since the operator-norm bound \er{eq:opnormresult} applies uniformly on the truncated large code, it directly gives the truncated analogue of the relation \er{eq:modflowisJLMSmixed} for a mixed target state.

The same idea also works, in either the untruncated or truncated setting, when the prescribed boundary reference state $\td\r$ is an appropriate mixed state. Suppose that $\td\r$ is at least approximately in the code subspace, i.e., it is close to the encoded image of a state on $\Hlarge$. Then as long as the code contains enough states to represent the support of $\td\r$, we may use a generalized modification map similar to the one above to change the encoding so that $\td\r$ is exactly an encoded state. The argument goes through as long as the modified encoding map satisfies the assumptions of the relevant exponentiated-JLMS theorem.

One may also ask whether the JLMS flow in the truncated code is close, as a bulk vector, to the original untruncated JLMS flow.  Writing $M:=\bigoplus_\a M_\a$, the question is whether
\begin{equation}
\label{eq:truncflowcomparison}
M e^{-is\left(\frac{\hat A_{cg,{\trunc}}}{4G}+K_{\rho_{{\trunc},r}}\right)}
|\phi_{\trunc}\rangle
\approx
e^{-is\left(\frac{\hat A_{cg}}{4G}+K_{\rho_r}\right)}|\phi\rangle
\end{equation}
holds. This need not hold for a generic truncation, since $x^{is}$ is sensitive to changes near $x=0$.  A convenient choice that allows \er{eq:truncflowcomparison} to hold is to truncate in the Schmidt basis of $|\psi^\a\rangle$, retaining matched Schmidt subspaces on $r$ and $\bar r$, taking each $|\psi^\a_{\trunc}\rangle$ to be the normalized projection of $|\psi^\a\rangle$ onto $\cH^\a_{\trunc}$, and taking $|\phi_{\trunc}\rangle$ to be the normalized projection of $|\phi\rangle$ onto $\cH_{\textit{large},{\trunc}}$.
If $P^\a_{r,{\trunc}}$ denotes the orthogonal projection onto $\cH^\a_{r,{\trunc}}$, the corresponding retained probability $p_{\trunc\mid\a}$ and normalized reduced state are given by
\begin{equation}
p_{\trunc\mid\a}:=\tr_r(P^\a_{r,{\trunc}}\rho_r^\a),
\qquad
\rho_{{\trunc},r}^\a
=\frac{P^\a_{r,{\trunc}}\rho_r^\a P^\a_{r,{\trunc}}}{p_{\trunc\mid\a}}.
\end{equation}
On $\cH^\a_{r,{\trunc}}$, one then has
\begin{equation}
(\rho_{{\trunc},r}^\a)^{is}P^\a_{r,{\trunc}}
=p_{\trunc\mid\a}^{-is}P^\a_{r,{\trunc}}(\rho_r^\a)^{is}.
\end{equation}
Thus, in each $\a$-sector, the two modular-flow operators differ only by the phase $p_{\trunc\mid\a}^{-is}$ on $\cH^\a_{r,{\trunc}}$.  Note that for this choice of truncation, $\|M_\a-\iota_\a\|_\infty=\bigl\||\psi^\a\rangle-|\psi^\a_{\trunc}\rangle\bigr\|=\sqrt{2(1-\sqrt{p_{\trunc\mid\a}})}$.
At fixed $s$, if $1-p_{\trunc\mid\a}$ is uniformly small and the projection onto $\cH_{\textit{large},{\trunc}}$ discards only a small-norm component of $|\phi\rangle$, then these phases are close to one, $M$ is close to the natural inclusion, and \eqref{eq:truncflowcomparison} follows.  Note that these are sufficient conditions for \eqref{eq:truncflowcomparison} to hold, not additional assumptions of the exponentiated-JLMS theorem required for \eqref{eq:opnormresult}.

In Appendix~\ref{app:sectorcoarsegraining}, we describe a related but different reduction of the large code Hilbert space by further coarse-graining groups of nearby $\a$-sectors. In that context, the coarse-graining again changes the bulk states by a small amount, and we can use a similar modification map to correct the error to ensure that the original boundary state $|\tilde\psi\rangle$ remains exactly in the code subspace.

\section{Large contributions from the bulk modular Hamiltonian}
\label{sec:largeKbulk}

Large but rare fluctuations about a classical background are as much a part of the bulk quantum state as are typical fluctuations.  In the large code, their contribution to the bulk modular Hamiltonian is represented by the exact central term $(-\log p_\a)$ in \eqref{eq:Kbulkblocks}, in addition to the modular Hamiltonian $K_{\rho_r^\a}$ within each $\a$-sector.

For a semiclassical reference state, one expects large fluctuations to be exponentially suppressed in $1/G$. The $(-\log p_\a)$ contribution to the bulk modular Hamiltonian can then be of order $1/G$ and thus naturally compete with $\hat A_{cg}/4G$ in the JLMS flow.  In a regime where the suppression of large fluctuations is Gaussian in the leading semiclassical approximation, this contribution grows quadratically with the fluctuation $\delta A:=A-\langle A\rangle$ with a coefficient of order $1/G$ and can become much larger than $\hat A_{cg}/4G$.

A particularly clear example arises for two-boundary AdS-Schwarzschild black holes in the approximation that we neglect local bulk excitations, boundary gravitons, angular momentum, and any gauge charges.  It also arises in pure JT gravity, where such approximations are unnecessary.  Up to diffeomorphisms, there is then only a one-parameter family of such solutions labeled by the horizon area $A$.  However, not all diffeomorphisms acting on these geometries are gauge symmetries.  As a result, the solutions define a two-parameter phase space labelled both by $A$ and $\delta$, the relative time-shift between the two asymptotic boundaries.  Here we use the time-shift terminology of the Jackiw-Teitelboim (JT) gravity discussion in Ref.~\cite{Harlow:2018tqv}, though the result is true in general (see e.g.\ Refs.~\cite{Thiemann:1992jj,Kastrup:1993br,Kuchar:1994zk}).\footnote{Instead of $\delta$, it is often natural to consider the canonical conjugate to $A$ which (up to normalization) is $\eta$, the relative boost angle between partial Cauchy slices on either side of the bifurcate horizon mentioned in section~\ref{sec:overview}.} 

One may alternatively label the classical phase space by the energy $E$ and the above $\delta$, normalizing $\delta$ so that these variables are canonically conjugate. Since bulk excitations have been forbidden, the gravitational Gauss law requires the ADM energies associated with the two boundaries to agree. As a result, this $E$ can be taken to be the common ADM energy. In this context it is useful to note that, for such solutions, the horizon area $A$ can be expressed as a (theory- and dimension-dependent) function $A(E)$ that does not depend on $\delta$.  Below, we use $|E\rangle_{\text{bulk}}$ to denote eigenstates of this energy.

Let us now consider the thermofield double state
\begin{equation}
\label{eq:TFD}
|\text{TFD}\rangle = Z^{-1/2}_\beta \sum_E e^{-\beta E/2}|E\rangle_R |E\rangle_L,
\end{equation}
in the tensor product CFT$_L \ \otimes$ CFT$_R$ of left ($L$) and right ($R$) CFTs dual to some bulk theory asymptotic to global AdS$_{d+1}$.   Here $Z_\beta$ is the CFT partition function at inverse temperature $\beta$, chosen so that the black-hole saddle dominates in the bulk. As emphasized in \cite{Maldacena:2001kr}, the bulk dual of \er{eq:TFD} is defined by a gravitational path integral in which the asymptotic boundary is taken to be a `half-torus' $[0,1] \times S^{d-1}$ of length $\beta/2$; see figure~\ref{fig:TFD}.

\begin{figure}[h!]
       \centering
        \includegraphics[width=0.3
        \textwidth]{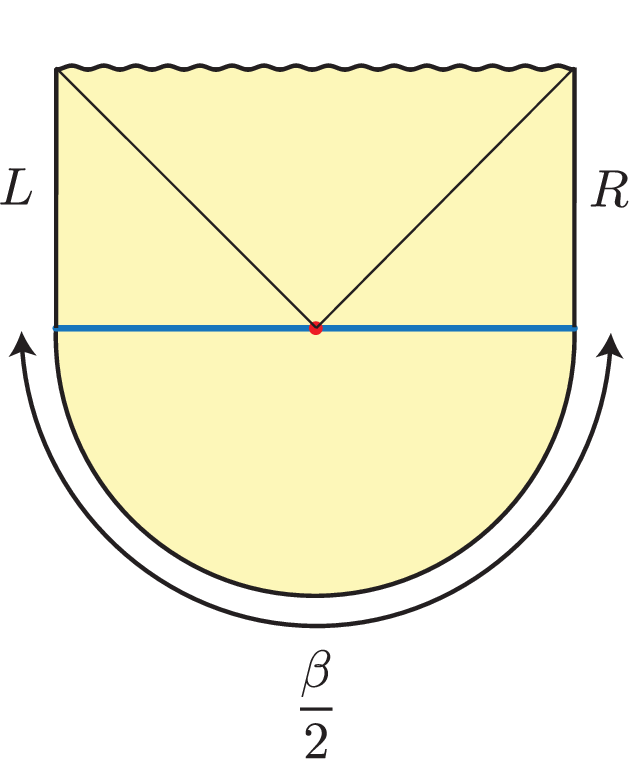}
        \caption{The thermofield double state is prepared in the bulk by a gravitational path integral with a `half-torus' asymptotic  boundary of length $\beta/2$. The dominant Euclidean saddle can be analytically continued at the moment of time symmetry (blue) to a Lorentzian two-sided black hole solution with asymptotic boundaries $L$ and $R$.}
        \label{fig:TFD}
\end{figure}

In terms of the CFT Hamiltonian $\hat H$,  we clearly have
\begin{equation}
\label{eq:KRTFD}
K_R = \beta \hat H + \log Z_\beta,
\end{equation}
so that \eqref{eq:UPJLMS} gives
\begin{equation}
\label{eq:KbulkTFD}
K_{\rho_r}
= \beta \hat H - \frac{\hat A}{4G} + \log Z_\beta,
\end{equation}
Here the HRT surface is the black hole horizon.  In writing \eqref{eq:KbulkTFD}, we remind the reader that throughout this work we consistently drop corrections at two-loop order and higher, and that such corrections were also dropped in writing \eqref{eq:UPJLMS}.

\begin{figure}[t]
      \centering
        \includegraphics[width=0.5
        \textwidth]{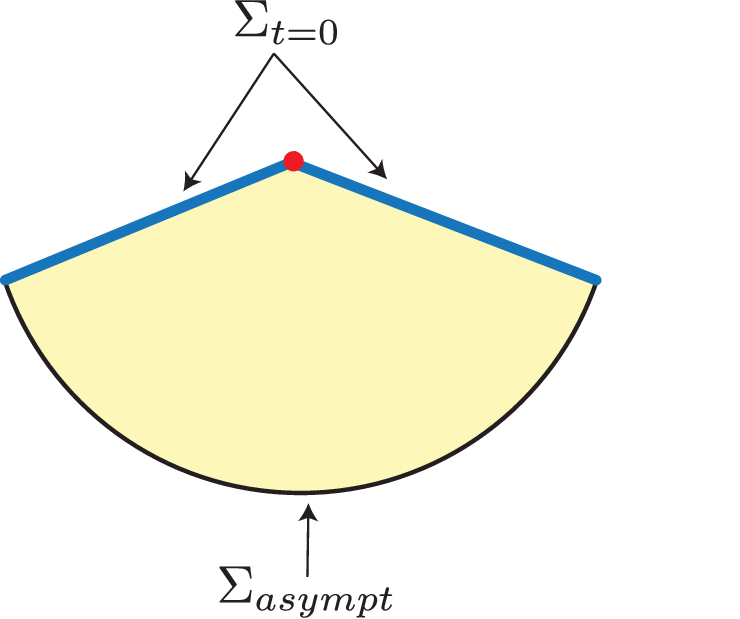}
        \caption{The wavefunction of the thermofield double state in the energy basis can be computed using a path integral with a half-torus asymptotic boundary of length $\beta/2$, denoted $\Sigma_{asympt}$, and an additional finite-distance boundary $\Sigma_{t=0}$ (blue). The induced metric on $\Sigma_{t=0}$ agrees with that on a constant-Killing-time slice of the AdS-Kruskal solution of given energy $E$. The saddle geometry generically involves a corner (red).}
        \label{fig:corner}
\end{figure}

It is important to note that \eqref{eq:KbulkTFD} can be reproduced by a direct path integral calculation in the bulk.  The individual terms in \eqref{eq:TFD} are characterized by a single number $E$ that gives the eigenvalue of the energy on both the left and right sides.  This is precisely the structure of the energy eigenstates associated with the quantization of the classical phase space truncated to the excitation-free two-sided Kruskal black holes described above.  As a result, after truncating to the above bulk phase space, the corresponding bulk wavefunction must be simply
\begin{equation}
\label{eq:TFDbulk}
|\text{TFD}\rangle_{\text{bulk}} = \int dE \ \mu(E) |E\rangle_{\text{bulk}}.
\end{equation}
The factor $\mu(E)$ may then be computed from the dominant saddle of the Euclidean bulk path integral with the same half-torus asymptotic boundary $\Sigma_{asympt}$ and an additional finite-distance boundary $\Sigma_{t=0}$ \cite{Harlow:2018tqv,Chua:2023srl}.  The induced metric on $\Sigma_{t=0}$ agrees with that on a constant-Killing-time slice of the two-boundary Schwarzschild-AdS solution of a given energy $E$; see \figref{fig:corner}.

The two boundaries $\Sigma_{asympt}$ and $\Sigma_{t=0}$ are required to meet in the usual way.  Notice, however, that the boundary conditions do not specify the extrinsic curvature of the surface $\Sigma_{t=0}$.  In particular, they allow saddle-point geometries in which $\Sigma_{t=0}$ has a `corner' at the Euclidean horizon as shown in \figref{fig:corner}. To compute $|\mu(E)|^2$, we glue this saddle to its conjugate. Recall now (see e.g.\ \cite{Carlip:1993sa}) that the Euclidean action $I_E$ of the resulting (doubled) geometry can be written in the canonical form
\begin{equation}\la{eq:action}
I_E = \int dt_E \left[\int d^dx \left(\dot{h}_{ij} \tilde \pi^{ij} +  N {\cal H} + N^i {\cal H}_i \right) + E \right] - \frac{A}{4G},
\end{equation}
in terms of the induced metric $h_{ij}$ on slices of constant Euclidean time $t_E$ and the (densitized) conjugate momentum $\tilde \pi^{ij}$. Since the constraints ${\cal H}, {\cal H}^i$ vanish on-shell and $\dot{h}_{ij}=0$ for static solutions, the action \er{eq:action} reduces to
\begin{equation}
I_E = \beta E- \frac{A}{4G}.
\end{equation}
As a result, at leading semiclassical order we find
\begin{equation}
\label{eq:bulkmu}
-\log |\mu(E)|^2 = \beta E - \frac{A}{4G} + \log Z_\beta,
\end{equation}
where the $\log Z_\beta$ term gives the correct normalization for the state \eqref{eq:TFDbulk}, as can be seen from \er{eq:TFD}.
See e.g.\ \cite{Harlow:2018tqv} and \cite{Marolf:2020vsi} for more detailed discussion of this result in the contexts of JT gravity and $(2+1)$-dimensional Einstein-Hilbert gravity (ignoring boundary gravitons).

The result \eqref{eq:bulkmu} reproduces \eqref{eq:KbulkTFD} as promised previously. To see this, we may view the bulk state \eqref{eq:TFDbulk} as a state in the ``direct-integral'' bulk Hilbert space $\int^\oplus dE\,\cH_{E,\text{bulk}}$,\footnote{This direct-integral Hilbert space is the continuous limit of the direct-sum $\Hlarge$ constructed in section~\ref{subsec:general}.  To follow the construction there strictly, we would choose a discrete set of area (or energy) windows \eqref{eq:areaWindows}; then $\Hlarge$ is a direct sum over these windows, with each summand spanned by the projection of \eqref{eq:TFDbulk} onto that window.  The resulting $\rho_r$ and $K_{\rho_r}$ would be coarse-grained versions of the expressions above.}
where each $\cH_{E,\text{bulk}}$ is a one-dimensional Hilbert space spanned by $|E\rangle_{\text{bulk}}$.  The bulk algebra $\cA_r$ consists of operators that act on $|E\rangle_{\text{bulk}}$ by multiplication with a function of $E$, i.e., all such operators are diagonal in the energy basis. The reduced state $\rho_r$, defined with respect to $\cA_r$ (and hence similarly diagonal in the energy basis), therefore acts by multiplication with $|\mu(E)|^2$ and is normalized with respect to the natural trace on $\cA_r$, $\tr_r\rho_r:=\int dE\,|\mu(E)|^2=1$, since then $\tr_r(\rho_r O)$ reproduces the expectation value of any $O\in\cA_r$ in the bulk state \eqref{eq:TFDbulk}. We emphasize that this trace is different from the Hilbert-space trace on the bulk Hilbert space. The modular Hamiltonian $K_{\rho_r}$ then acts by multiplication with $-\log|\mu(E)|^2$.  Equation~\eqref{eq:bulkmu} therefore gives \eqref{eq:KbulkTFD}.

Let us illustrate the important features of such examples by considering a nonrotating BTZ black hole.  Setting the AdS length to one, taking our thermofield double to be peaked at a black hole of area $A_0$ (so that  $\beta = 4\pi^2/A_0$), and using $E=A^2/(32\pi^2G)$,
we find that \eqref{eq:KRTFD} and \eqref{eq:KbulkTFD} take the form \cite{Marolf:2020vsi}
\begin{align}\la{eq:KRBTZ}
        K_R &= \frac{\hat A^2}{8 G A_0}+\log Z_\beta, \\
        K_{\rho_r} &= K_R-\frac{\hat A}{4G} = \frac{(\hat A-A_0)^2}{8 G A_0}+\log Z_\beta-\frac{A_0}{8G},\la{eq:KrhorBTZ}
\end{align}
which are valid when acting on states in the code subspace consisting of BTZ black holes of varying horizon areas. Here, $A_0 = 4\pi^2/\beta$ is the expectation value of $\hat A$ in the thermofield double state \er{eq:TFD} defined with a particular $\b$.\footnote{In 2+1 dimensions, $A$ is in fact the \textit{length} of the horizon, but we continue to refer to it as the horizon area.}  These expressions make manifest that $K_{\rho_r}$  can be larger than $\hat A/4G$ when acting on states with large $A/A_0$, but that its action becomes less important near $A=A_0$.

In the above discussion, we neglected all degrees of freedom other than the background geometry for simplicity.
But one may repeat the discussion in the context of a general large code of section~\ref{sec:JLMSflow}. In such a context, the bulk modular Hamiltonian $K_{\r_r}$ is given by \eqref{eq:Kbulkblocks}, which in particular contains a central part $(-\log p_\a)$. Focusing on this central part, the discussion then proceeds in direct parallel with the simple BTZ example above. Thus, as before, (the central part of) $K_{\r_r}$ can compete with or dominate over the $\hat A_{cg}/4G$ term when acting on target states whose restricted classical background is different from that of the reference state $\r$.

On the other hand, when the reference and target states have the same restricted classical background, the situation is different. In this case, even though the $K_{\rho_r^\a}$ term in the bulk modular Hamiltonian \eqref{eq:Kbulkblocks} generically generates nontrivial modular flow within an $\a$-sector (in the sense of causing order-one changes in the quantum state), we will nevertheless show in the next section that, under appropriate conditions, neither this $K_{\rho_r^\a}$ term nor the central part $(-\log p_\a)$ affects the target state's (complete) classical background under modular flow, so its leading-order change is given by the area flow alone.

\section{The area-flow description of boundary modular flow}
% (fold)
\label{sec:KandA}

We continue to consider the reference and target states $\td\r$ and $\td\t$ introduced in section~\ref{sec:JLMSflow}, under the conditions assumed there for boundary modular flow to be approximated by the JLMS flow. However, we will now seek additional conditions under which boundary modular flow defined by $\td\r$ is semiclassically well approximated by the HRT-area flow when acting on the target state $\td\t$.

One motivation for this study is that the conclusions of section~\ref{sec:largeKbulk} above may appear to be in tension with the results of Ref.~\cite{Bousso:2020yxi}. In particular, the previous section showed that, for a general target state concentrated far from the peak of the reference state, the central $(-\log p_\a)$ part of $K_{\rho_r}$ can compete with or dominate over $\hat A_{cg}/4G$. By contrast, Ref.~\cite{Bousso:2020yxi} found strong evidence that the effect of boundary modular flow on the classical bulk background is described by the boundary-condition-preserving kink transformation generated by $A/4G$.\footnote{Ref.~\cite{Bousso:2020yxi} focused on Connes cocycle (CC) flow, which in principle differs from modular flow by an additional flow generated by the vacuum modular Hamiltonian (and which improves the UV behavior in contexts where the HRT surface reaches the asymptotically AdS boundary).  However, in the cases studied there, this vacuum modular flow acted as a geometric symmetry which is easily undone, so the results can be read as evidence relating modular flow to area flow.}
However, one should also realize that, using our terminology, Ref.~\cite{Bousso:2020yxi} considered only cases in which the reference and target states are identical, so that their restricted classical backgrounds $\bar\Phi_\rho$ and $\bar\Phi_\t$ coincide.

We may thus resolve the above tension by showing more generally that, under the conditions below, agreement of the restricted backgrounds is sufficient for the flow generated by $K_{R,\psi}$ on the complete target-state background $\Phi_\t$ to be well approximated by the HRT-area flow. We emphasize that the complete backgrounds $\Phi_\rho$ and $\Phi_\t$ need not coincide in the regions to the future and past of the HRT surface. We remind the reader that our conventions for complete and restricted classical backgrounds were stated in section~\ref{sec:overview}.

For the reference state, we use $p_\a$ and $\r^\a$ as defined by \er{rdecom}. We similarly write the diagonal blocks of the target state as
\begin{equation}\la{tdecom}
\t^{\a\a}=q_\a\t^\a,\qquad
q_\a := \tr \t^{\a\a}.
\end{equation}
Here $q_\a$ is the target-state weight of the $\a$-sector.

For simplicity, throughout this section we take the area windows to have an equal width that is $o(G)$, except for a set of windows carrying small total target-state weight.\footnote{This equal-width assumption can be relaxed. For unequal widths, the argument in the rest of this section goes through with $A_{cg}^\a$ and $p_\a$ replaced by $A_\a$ and $p_\a/\e_\a$, respectively. In particular, statements about $K_{\rho_r}=\bigoplus_{\a}[(-\log p_\a)\mathds{1}_r^\a +K_{\rho_r^\a}]$ are understood as statements about $\bigoplus_{\a}[(-\log \frac{p_\a}{\e_\a})\mathds{1}_r^\a +K_{\rho_r^\a}]$ instead.}

The argument is then rather simple.  It is useful to start by recalling
that our present formalism splits the notion of a classical background $\Phi$ into two pieces.  The first piece we call the restricted classical background $\bar\Phi$, which determines the classical values of observables that Poisson-commute with the HRT area.  Such a restricted background is determined by the label $\a$ assigned to each small code.  This label $\a$ also specifies an HRT area, but information about quantities that do not Poisson-commute with the HRT area is encoded in the relative phases of states assigned to small codes with different $\a$ (but which still correspond to the same restricted background $\bar\Phi$).  Controlling the action of JLMS flow on the classical background is thus a matter of controlling the action of this flow on the probabilities and appropriate phases the flowed state assigns to each small code.

With this in mind, let us recall that, with the normalization described in section~\ref{sec:overview}, our field data have Poisson brackets of order $G$.  For example, the Poisson bracket of the local boost $\eta$ with $A$ is of order $G$, since $\eta$ is canonically conjugate to $A/4G$.  As a result, phases that are $o(G^{-1})$ do not affect the resulting classical background under modular flow. Here we assume that such phases vary sufficiently regularly with $A$ that differentiation with respect to $A$ preserves their $o(G^{-1})$ scaling. We will make a similar assumption for the estimates involving $K_{\rho_r^\a}$ given below.

We first consider the $K_{\rho_r^\a}$ term in \eqref{eq:Kbulkblocks}, which generates modular flow within each $\a$-sector.  As a result, it clearly leaves invariant the restricted classical background $\bar\Phi_\t$ associated with the target state. We now argue that such flows also have no leading-order effect on the phases that specify the remainder of the target state's complete classical background.

In the truncated setting of section~\ref{subsec:truncated}, this follows from log-stability with a reasonably small $\e_{tail}$:
\begin{equation}
\lt\|K_{\rho_r^\a}\rt\|_\infty
\leq |\log\epsilon_{tail}|=o(1/G),
\end{equation}
where we have dropped a `trunc' subscript for simplicity.
The above condition ensures that $K_{\rho_r^\a}$ cannot produce a phase of order $G^{-1}$.  The restriction on $\epsilon_{tail}$ is rather weak: $\epsilon_{tail}$ may tend to zero as $G\to 0$, provided that $|\log\epsilon_{tail}|$ grows more slowly than $1/G$.

In the untruncated setting of section~\ref{subsec:general}, the promised control over the leading-order effect of $K_{\rho_r^\a}$ follows from
\begin{equation}\la{eq:Krrasubleadtau}
\sum_\a q_\a \< K_{\rho_r^\a} \>_{\t_r^\a} \lesssim \sum_\a q_\a\big[S(\t_r^\a)+|\log\epsilon_{tail}|\big] \leq S(\t_r)+|\log\epsilon_{tail}| =o(1/G),
\end{equation}
where the first inequality follows from relative log-stability,\footnote{Relative log-stability with respect to the four test states in \eqref{eq:statecombine} implies relative log-stability with respect to the state $\t^\a$, up to an inessential order-one reduction of $\epsilon_{tail}$. The first inequality then follows as in the proof of Theorem~1b of Ref.~\cite{JLMS}.}
and in the last step we imposed a mild condition on the bulk entropy $S(\t_r)$ (in addition to the same weak restriction on $\epsilon_{tail}$ as above). To see that \er{eq:Krrasubleadtau} is sufficient to control the effect of $K_{\rho_r^\a}$ as promised, we note that it ensures, at leading order, that the target state $\t$ does not probe eigenvalues of $K_{\rho_r^\a}$ that are of order $1/G$. In particular, for any fixed $c>0$, the projector onto the spectral subspace $K_{\rho_r^\a}\geq c/G$ is bounded by $G K_{\rho_r^\a}/c$, so the total weight of the target state $\t$ on such spectral subspaces is bounded by the expectation value of $\oplus_\a G K_{\rho_r^\a}/c$, which is only $o(1)$. Thus the target state -- up to this $o(1)$ tail -- does not probe $K_{\rho_r^\a}$ eigenvalues of order $1/G$.

Thus, in either setting, the modular flow generated by $K_{\rho_r^\a}$ within each $\a$-sector has no leading-order effect on the target state's classical background.
It follows that only the central $(-\log p_\a)$ term in \eqref{eq:Kbulkblocks} can affect the target state's classical background at order $G^{-1}$.  Its effect is determined by the probabilities $p_\a$ defined by the reference state and the probabilities $q_\a$ defined by the target state.

Let us suppose that the probabilities $p_\a$ defined by the reference state $\r$ are well approximated by a Gaussian distribution with standard deviation $\Delta_\rho$. More generally, when $\alpha$ describes more than a single parameter, the right object to discuss is the covariance matrix describing potentially different standard deviations in different directions.  Here we use one-dimensional language for simplicity, since every width has the same relevant $G^{1/2}$ scaling, and also because we have previously emphasized the one-parameter case. We take this standard deviation to be defined using an appropriate notion of distance $|\a-\a'|$ between sectors $\a, \a'$ defined by comparing the classical fields in the restricted backgrounds $\bar\Phi_\a,\bar\Phi_{\a'}$ without introducing any additional factors of $G$.  For example, when $\alpha$ is just an area-window label, we may simply take $|\a-\a'|=|A_\a-A_{\a'}|$.

In particular, supposing that the semiclassical state $\rho$ is peaked around sector $\alpha_\rho$, the Gaussian approximation to its probabilities $p_\a$ takes the form
\begin{equation}
\label{eq:gaussian}
\log\left(\frac{p_\a}{p_{\a_\rho}}\right) = -\frac{|\a-\a_\rho|^2}{2\Delta_\rho^2}\bigl[1+o(G^0)\bigr]
\end{equation}
within the regime
\begin{equation}
\label{eq:gaussregime}
\frac{|\a-\a_\rho|}{\Delta_\rho} = o(G^{-1/2}).
\end{equation}
Here the non-Gaussian corrections are assumed to vary sufficiently slowly with the classical fields that their Poisson brackets with the field data remain $o(G^0)$ throughout this regime.
Furthermore, since $\rho$ and $\t$ describe the same restricted classical background as $G\rightarrow 0$, for every sector $\a$ outside a set carrying negligible total target-state weight, the distance defined above gives $|\a-\a_\rho|\rightarrow 0$ as $G\rightarrow 0$.  Such sectors thus lie in the regime \eqref{eq:gaussregime}, and the associated probabilities $p_\a$ hence satisfy
\begin{equation}
-\log p_\a + \log p_{\a_\rho} = \log \(\frac{p_{\a_\rho}}{p_\a}\) \ap \frac{|\a-\a_\rho|^2}{2\Delta_\rho^2} = o(1/G).
\end{equation}

We thus find that up to an $\a$-independent part that produces only an overall phase, $(-\log p_\a)$ is negligible at order $1/G$ (throughout these sectors which carry nearly all the target-state weight) and its Poisson brackets with the field data are $o(G^0)$, so it generates no leading-order change in the classical background. Together with the control over $K_{\rho_r^\a}$ established above, this shows that at order $1/G$, modular flow by $K_{\rho_r}$ has no effect on the complete target-state background $\Phi_\t$.  As a result, under the JLMS flow the leading-order change in the background $\Phi_\t$ is generated entirely by $\hat A_{cg}/4G$.

Moreover, since the area windows were taken to have equal width (except for a set of windows carrying small total target-state weight), \er{eq:Aalphacg} then shows that, up to a small error, replacing $A_{cg}^\a$ by $A_\a$ changes this area flow only by an overall phase.  Moreover, on these windows, $A_\a$ differs from the geometric HRT area $A$ by at most half the window width, which is $o(G)$.  The unitary-flow identity \eqref{eq:unitaryflowdifference} therefore allows us to replace $A_\a$ by the geometric $A$ in this area flow, again up to an overall phase and a small error.  The leading-order change in the background $\Phi_\t$ is thus generated by $\hat A/4G$ and coincides with the boundary-condition-preserving kink transformation (see also figure~\ref{fig:areaflow} in section~\ref{sec:overview}).  

In particular, the state
\begin{equation}
\td\t_s
:=\tilde{\rho}_R^{is}\td\t\tilde{\rho}_R^{-is}
=e^{-isK_{R,\psi}}\td\t e^{isK_{R,\psi}}
\end{equation}
resulting from the flow will again be semiclassical in the sense described above for order-one values of the flow parameter $s$.
However, due to increasing dispersion in $\eta$ generated by the quadratic-in-$A$ term in $K_{\rho_r}$, the state $\td\t_s$ will generally cease to be semiclassical at parametrically large values of $s$.

Let us now return to the results of~\cite{Bousso:2020yxi}.  In that work, the reference and target states are identical, so their restricted classical backgrounds $\bar\Phi_\rho$ and $\bar\Phi_\t$ coincide. Our argument then shows that $\td\t_s$ is again semiclassical at order-one values of $s$ and describes a classical background $\Phi_{\t,s}$ related to $\Phi_\t$ by a boundary-condition-preserving kink-transformation.

We should also comment briefly on quantum corrections.  As already noted, the bulk modular Hamiltonian generally induces significant changes in the bulk quantum state.  A particularly important effect occurs near the HRT surface.   Since the area flow creates a kink while leaving the quantum state unaffected,  the action of the area-flow alone on the quantum state would destroy the local vacuum structure near the HRT surface.  We therefore remind the reader that, as argued in Ref.~\cite{Jafferis:2015del}, the action of the bulk modular Hamiltonian $K_{\rho_r}$ precisely repairs this local vacuum\footnote{It is perhaps also useful to note that 
the bulk modular flow also induces significant changes in the quantum state away from the HRT surface.
Even when acting on the original state $\rho$, such changes can be 
detected by two-sided correlation functions. In contrast, since the area operator commutes with bulk operators in either wedge, the area flow leaves such correlation functions invariant.}.
However, since we impose a UV cutoff on the bulk fields, we also note that the energy density associated with the destruction of the vacuum by the area-only flow is a dimension-dependent power of $\Lambda$. As a result, as long as we choose $\Lambda\ll l_{P}^{-1}$ where $l_P$ is the Planck length, the backreaction of this singularity can be consistently neglected at leading order in $1/G$. In this sense the bulk classical background obtained from bulk area-only flow continues to agree with that obtained from modular flow in the CFT even when quantum corrections are included.

\section{Discussion}
\label{sec:discussion}

The goal of this note was to clarify certain points regarding the use of the JLMS formula for semiclassical bulk states.  The large-code formalism of Ref.~\cite{JLMS} provides controlled forms of the exponentiated JLMS relation for appropriate semiclassical states in contexts that go beyond small perturbations around a single classical background.  We first used this formalism to identify conditions under which boundary modular flow is approximated by the JLMS flow.

We then asked when the leading-order effect of the JLMS flow on the classical background is generated by the area term $\hat A_{cg}/4G$.
In the above setting, the bulk modular Hamiltonian $K_{\r_r}$ is generally {\it not} negligible in comparison with the area term $\hat A_{cg}/4G$; in particular, the central $(-\log p_\a)$ part of $K_{\r_r}$ can be of order $1/G$.  Nevertheless, when the reference and target states have the same restricted classical background, and when the mild assumptions of section~\ref{sec:KandA} hold, the bulk modular flow has no leading-order effect on the complete classical background of the target state.  The boundary modular flow therefore acts on that background as the flow generated by $A/4G$, namely the BCP kink transform.

We close with several applications, extensions, and open directions.

\paragraph{Commutators of modular Hamiltonians:}
Recall that recent works introduced commutators of modular Hamiltonians as an entanglement measure for gapped $(2+1)$-dimensional systems \cite{Kim:2021gjx,Kim:2021tse} and subsequently studied them in $(1+1)$-dimensional CFTs and holography \cite{Zou:2022nuj}. In particular, the quantity studied was of the form
\begin{equation}
\label{eq:modcom}
    J(A,B,C)_\Omega:=i\langle [K_{AB},K_{BC}]\rangle_\Omega
\end{equation}
where the subscript $\Omega$ indicates that this expectation value is taken in the vacuum $|\Omega\rangle$.  The modular Hamiltonians are defined by this same vacuum state restricted either to the region $AB$ or to the region $BC$, where $A,B,C$ are contiguous intervals at some given time.

Let $S_{BC}(s)$ denote the entropy of $BC$ in the state $e^{-isK_{AB}}|\Omega\rangle$.  The first law of entanglement gives
\begin{equation}
J(A,B,C)_\Omega = \left.\frac{dS_{BC}(s)}{ds}\right|_{s=0}.
\end{equation}
In a holographic context, we may use our result to determine the complete classical background of the state $e^{-isK_{AB}}|\Omega\rangle$, and hence the HRT area for $BC$ that determines the leading $O(1/G)$ part of $S_{BC}(s)$.
Provided that the small-$G$ approximation is sufficiently uniform near $s=0$ so that the small-$G$ and small-$s$ limits may be interchanged, \eqref{eq:modcom} reduces at leading order to the change of $A_{BC}/4G$ under the flow generated by $A_{AB}/4G$, which is given by $i$ times the expectation value of the commutator of the HRT area operators $\hat A_{AB}/4G$ and $\hat A_{BC}/4G$ (or the appropriate generalization in theories with higher-derivative corrections).  This allows the results of \cite{Kim:2021gjx,Kim:2021tse,Zou:2022nuj} to be compared directly with the area-commutators (and generalizations thereof) computed in \cite{Kaplan:2022orm,Kaplan:2023oaj}.

\paragraph{Multiple semiclassical peaks:}
We focused on situations involving semiclassical states with a single classical background, but our results extend readily to cases with multiple semiclassical peaks.
Consider, for example, a reference state whose probabilities $p_\a$ are peaked at an $O(G^0)$ number of locations corresponding to well-separated restricted classical backgrounds $\bar\Phi_i$; such a state could be obtained as a superposition or mixture of semiclassical states with distinct classical backgrounds.
We may then repeat the construction of section~\re{sec:JLMSflow} and the argument of section~\re{sec:KandA}, finding that when the target state has a restricted classical background that coincides with one of the $\bar\Phi_i$, its complete classical background is given by the corresponding BCP kink transform.
Moreover, if the target state is a superposition or mixture of semiclassical states $\t_j$, each of which has a restricted classical background that coincides with one of the $\bar\Phi_i$, then the boundary modular flow applies the BCP kink transform to the complete classical background of each $\t_j$.

\paragraph{States without a single entanglement wedge:}
It is interesting to extend our results to situations studied in Ref.~\cite{Akers:2020pmf} which do not have a single well-defined entanglement wedge.
One such situation involves a ball of dust in a mixture of a pure state $\rho_1$ and a highly mixed state $\rho_2$ (see \figref{fig:incom}). The amount of dustball entropy in $\r_2$ was chosen so that the dominant quantum extremal surface (QES) for the two-interval region $R$ is the connected surface $\gamma_1$ in $\rho_1$ and the disconnected surface $\gamma_2$ in $\rho_2$. Although the two states have the same classical background, their mixture does not satisfy the naive QES prescription nor have a single well-defined entanglement wedge.

\begin{figure}[t]
       \centering
        \includegraphics[width=0.8
        \textwidth]{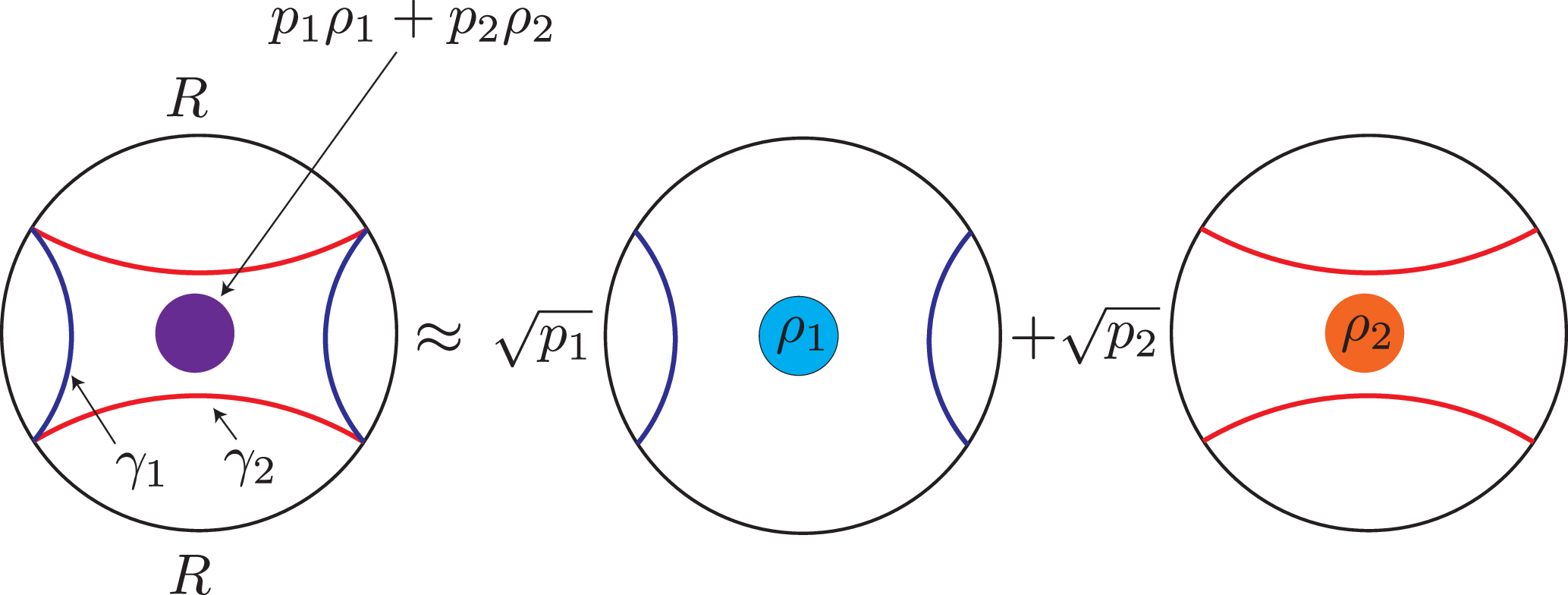}
        \caption{A ball of dust in AdS is chosen to be in a mixture of a pure state $\rho_1$ and a highly mixed state $\rho_2$. The two states have the same classical background, but the dominant QES for the two-interval region $R$ is $\gamma_1$ (blue) in $\rho_1$ and $\gamma_2$ (red) in $\rho_2$, so the mixture does not have a single well-defined entanglement wedge.}
        \label{fig:incom}
\end{figure}

As a warm-up, we first note that the construction of section~\re{sec:JLMSflow} and the result of section~\re{sec:KandA} apply directly to $\rho_1$ and $\rho_2$ separately, under the assumptions described there. For $\rho_2$, before using it as the reference state, it is useful to first purify it using an auxiliary system regarded as part of $\bar R$; alternatively, the purifier may be supplied by bulk degrees of freedom localized near $\bar R$. In the following, we will use $\r_2$ to denote (the density operator for) this purified state. The associated small codes are taken to contain the reference state $\r_2$ constrained to area windows for $\g_2$, together with perturbative excitations that leave unchanged the entangled state of the dustball and its purifier. The result of section~\re{sec:KandA} then holds: when a target state in the code has the same restricted classical background as $\r_2$, its complete classical background is given by the BCP kink transform at $\gamma_2$. The same result holds when using $\r_1$ as the reference state, except that the kink transform is applied at $\g_1$. Note that although $\rho_1$ and $\rho_2$ have the same complete classical background, the restricted backgrounds used in their corresponding code constructions are distinct because they are defined relative to $\gamma_1$ and $\gamma_2$, respectively.

Now we consider the situation where the reference state $\r$ is a superposition or mixture of the two states.
In this case, we slightly generalize the construction of section~\re{sec:JLMSflow} by taking the large code to be the direct sum of all small codes constructed around $\rho_1$ and all those constructed around $\rho_2$. Under the same assumptions as above, including approximate subsystem orthogonality between these small code subspaces, the result of section~\re{sec:KandA} again holds.
For example, when the target state is the reference state itself, the boundary modular flow is described at leading order by applying the BCP kink transform at $\gamma_i$ to each component associated with $\rho_i$. In particular, the two components associated with $\r_1, \r_2$ initially have the same complete classical background, but after some nonzero amount of boundary modular flow, the two components describe distinct kinked backgrounds.

\paragraph{Large bulk entropy and the IR area operator:}
While our study focused on situations involving classical extremal surfaces, large parts of our argument used only the large-code formalism of Ref.~\cite{JLMS} which holds in much greater generality. In particular, the formalism may be applied to contexts in which a QES defines the relevant entanglement wedge, regardless of whether the QES is a small perturbation of a classical extremal surface. This suggests applications of our results to settings such as late stages of black hole evaporation studied in \cite{Penington:2019npb,Almheiri:2019hni}.

In the example without a single entanglement wedge discussed above, the large dustball entropy changes which QES dominates, but it does not contribute at order $1/G$ to the generalized entropy of the dominant QES in either component: the dustball state is pure in $\rho_1$, while the dominant QES in $\rho_2$ excludes the dustball from the entanglement wedge.  We now focus on situations in which the bulk-entropy contribution to the generalized entropy of the dominant QES is itself of order $1/G$, so that its effective area can differ at leading order from its geometric area.  The large-code formalism may accommodate such situations by replacing the area operator $\hat A_{cg}$ with an IR (renormalized) area operator $\hat A_\IR$ that incorporates the large matter entropy.
In particular, Ref.~\cite{RG} constructed an appropriate large code in the IR that could satisfy the FLM and R\'enyi FLM formulas with such an $\hat A_\IR$.
Moreover, the generalized entropy can now be dominated by $\hat A_\IR/4G$ in such a code.
Under the same assumptions as above, the boundary modular flow is then described at leading order by the flow generated by $\hat A_\IR/4G$ rather than the geometric area.

It would be interesting to explore the flow generated by such an $\hat A_\IR/4G$ in future work. Indeed, a flow including stress-tensor shocks from the matter sector was briefly motivated for a QES in Ref.~\cite{Bousso:2020yxi}. Such an $\hat A_\IR/4G$ flow may also share features with the more general flows studied in Ref.~\cite{Dong:2025orj} which are generated by the geometric entropy in higher-derivative theories and can act nontrivially on matter as well as geometric initial data.

At the moment, we simply describe a toy example obtained by entangling $n(A)$ Bell pairs across $\gamma_R$, with $n(A)$ depending sufficiently smoothly on the geometric area $A$. Their contribution gives schematically
\begin{equation}
\frac{A_\IR}{4G} = \frac{A}{4G}+n(A)\log 2,\qquad
n(A)=n(A_0)+n'(A_0)(A-A_0)+\cdots,
\end{equation}
where $A_0$ is the peak of the probability distribution for $A$.  If $n(A_0)$ is of order $1/G$, then $A_\IR$ differs from $A$ at leading order.  Nevertheless, the constant term $n(A_0)$ produces only an overall phase and does not change the classical flow.  On the other hand, the linear term rescales the BCP kink transform by a factor $1+4G(\log 2)n'(A_0)$.  This has the same form as an effective renormalization of Newton's constant, as in Ref.~\cite{Susskind:1995qc}.

\paragraph{Bulk cutoffs and von Neumann algebras:}
Our use of the large-code formalism of Ref.~\cite{JLMS} required the use of a UV cutoff in the bulk as well as technical control from conditions such as log-stability.
It would be interesting to explore whether the language of type II and type III von Neumann algebras provides a cleaner and more elegant formulation of the results described here, perhaps along the lines of the discussion in Refs.~\cite{Faulkner:2022ada,Crann:2024gkv}.

\paragraph{More general modular flow:}
One of the points stressed above was that a general discussion of the bulk dual of modular flow at order-one flow parameters requires a notion of a `large code subspace' that allows states that correspond to distinct classical solutions in the limit $G \rightarrow 0$.    Our discussion here was based on the large codes constructed in \cite{JLMS}.  However, one should note that that formalism allows only a limited set of such classical backgrounds, as small codes built around an arbitrary collection of classical backgrounds need not satisfy the requirement of approximate subsystem orthogonality mentioned in section~\ref{sec:JLMSflow}. In our construction, we also restricted attention to target states that are supported almost entirely in the same perturbative Hilbert spaces as the given reference state, $\a$-sector by $\a$-sector.  It would thus be of great interest either to find a more general formalism that allows one to study more generic superpositions of classical backgrounds and then check whether the conclusions here remain valid, or to have a better understanding of why in principle such a formalism might not exist.  

%~~~~~~~~~~~~~~~~~~~~~~~~~~~~~~~~~~~~~~~~~~~~~~~~~~~~~~~~~~~~~~~~~~~~~
\acknowledgments

%~~~~~~~~~~~~~~~~~~~~~~~~~~~~~~~~~~~~~~~~~~~~~~~~~~~~~~~~~~~~~~~~~~~~~

This material is based upon work supported by the Air Force Office of Scientific Research under Award Number FA9550-19-1-0360. This material is also based upon work supported by the U.S. Department of Energy, Office of Science, Office of High Energy Physics, under Award Number DE-SC0011702. The work of DM and PR was also supported by a grant from the Simons Foundation.
This work was supported in part by the Leinweber Institute for Theoretical Physics; and by the Department of Energy, Office of Science, Office of High Energy Physics under Award DE-SC0025293. We also acknowledge support from the University of California. This work is supported by the Department of Atomic Energy, Government of India, under Project Identification Number RTI-4012.

\appendix

\section{Stability of the JLMS flow under coarse-graining of nearby $\alpha$-sectors}
\label{app:sectorcoarsegraining}

Let us start with a collection $\mathcal G$ of nearby $\a$-sectors (whose area windows are assumed, for simplicity, to have approximately equal widths). The relevant part of our large code is thus
\begin{equation}
\bigoplus_{\a \in \mathcal G} \cH_r^\a \otimes \cH_{\bar r}^\a.
\label{eq:Gnbas}
\end{equation}
We can coarse-grain \eqref{eq:Gnbas} into a new small code with a factorized Hilbert space (which we may call $\cH_r^{\mathcal G} \otimes \cH_{\bar r}^{\mathcal G}$) while making only a small change in the JLMS flow.
This can be understood as combining nearby area windows into a single larger window. To see that we can do this, note that as long as the new area window is still sufficiently small, we expect that the factors $\cH_r^\a$ for $\a$ in the group can be identified with a common Hilbert space $\cH_r^{\mathcal G}$ (possibly after truncations of the type discussed in section~\ref{subsec:truncated}), and similarly for $\cH_{\bar r}^\a$:
\begin{equation}
\cH_r^\a \cong \cH_r^{\mathcal G},\qqu
\cH_{\bar r}^\a \cong \cH_{\bar r}^{\mathcal G}.
\end{equation}
We also expect semiclassical states $\rho_r^\a, |\phi^\a\rangle$ and the coarse-grained area $A_{cg}^\a$ to vary across the group by small amounts $\Delta_\rho$, $\Delta_\phi$, and $\Delta_A$ controlled by the width of the new area window.  Our coarse-graining procedure is then to choose an arbitrary $\a$-sector as a representative, using the component of the original states $\rho_r^\a, |\phi^\a\rangle$ (and coarse-grained area $A_{cg}^\a$) in this chosen $\a$-sector as the corresponding states (and coarse-grained area) in the new small code $\cH_r^{\mathcal G} \otimes \cH_{\bar r}^{\mathcal G}$.
This procedure therefore amounts to replacing these data in every other $\a$-sector by nearby data.  The theorem below bounds the resulting change in the JLMS flow.\footnote{Since the total probability $p_{\mathcal G}$ on this group of windows must be preserved, it is generally larger than the original probability $p_\a$ for the $\a$-sector chosen as the representative. This difference in the resulting modular Hamiltonians is absorbed by `renormalizing' the more coarse-grained area $A_{cg}^{\mathcal G} = A_{cg}^\a + 4G \log (p_{\mathcal G}/ p_\a)$ with $\a$ denoting the chosen sector; the difference between this $p_\a$ and the $p_{\a'}$ for a nearby $\a'$-sector is small due to equal area-window width, and thus only causes a small error in the JLMS flow according to \er{eq:unitaryflowdifference}.}

Moreover, viewing the result as a relation between JLMS flows with more fine-grained area windows and more coarse-grained ones, we may try to take a limit where the widths of fine-grained area windows go to zero in order to recover the continuous geometric area operator $\hat A$ (while still keeping the coarse-grained windows appropriately small). Therefore, we may view the result here as evidence that this continuous limit of the JLMS flow is approximately equal to the coarse-grained JLMS flow, and due to \er{eq:modflowisJLMS}, it is also approximately equal to the boundary modular flow (a statement that we discussed in section~\ref{sec:JLMSflow}).

When defining an encoding map for each new small code $\cH_r^{\mathcal G} \otimes \cH_{\bar r}^{\mathcal G}$ and assembling these codes into a new large code, we may wish to ensure that the coarse-grained version of $|\psi\rangle$ still gets encoded as the original boundary state $|\tilde\psi\rangle$ (so that the boundary modular flow it defines remains unchanged). This can be achieved by using a modified encoding map that corrects the small error caused by the coarse-graining, similar to \er{eq:ModMa} and \er{eq:ModVa}.

\begin{theorem}[Sector coarse-graining bound]
\label{thm:sectorcoarsegraining}
For $\a=1,2$, let $\rho^\a$ be a state and $\t^\a=|\phi^\a\rangle\langle\phi^\a|$ be a pure state on $\cH_r\otimes\cH_{\bar r}$.  Suppose that $\rho^\a$ is relatively log-stable with respect to $\t^\a$ in the sense of \eqref{eq:relativelogstabilitybulk}, with a common $\epsilon_{tail}$.  Set
\begin{equation}
\label{eq:cgdistances}
\Delta_\rho:=\|\rho_r^1-\rho_r^2\|_1,
\qquad
\Delta_\phi:=\bigl\||\phi^1\rangle-|\phi^2\rangle\bigr\|,
\qquad
\Delta_A:=|A_{cg}^1-A_{cg}^2|,
\end{equation}
where $A_{cg}^\a$ is real.  For fixed $s\in\mathbb R$, define
\begin{equation}
\label{eq:cgflow}
U_\a(s):=e^{-is\fr{A_{cg}^\a}{4G}}(\rho_r^\a)^{is}.
\end{equation}
Then we have
\begin{equation}
\label{eq:cgflowreplace}
\Bigl\|U_1(s)|\phi^1\rangle-U_2(s)|\phi^2\rangle\Bigr\|
\leq
\Delta_\phi
+\sqrt{2C_s}\left(\frac{\Delta_\rho}{\epsilon_{tail}}\right)^{1/4}
+|s|\fr{\Delta_A}{4G},
\end{equation}
with $C_s:=1+\sqrt{1+4s^2}$.
\end{theorem}

\begin{proof}
Let $\varrho_\a:=\rho_r^\a$.  
We first show
\begin{equation}
\label{eq:cgoperatorcontinuity}
\left\|(\varrho_1^{is}-\varrho_2^{is})\varrho_1^{1/2}\right\|_2
\leq C_s\Delta_\rho^{1/2}.
\end{equation}
Here $\|X\|_p:=\left[\tr(|X|^p)\right]^{1/p}$ is the Schatten-$p$ norm, with $|X|:=(X^\dagger X)^{1/2}$.  To see \eqref{eq:cgoperatorcontinuity}, note that
\begin{align}
\left\|(\varrho_1^{is}-\varrho_2^{is})\varrho_1^{1/2}\right\|_2
&\leq \left\|(\varrho_1^{1/2})^{1+2is}-(\varrho_2^{1/2})^{1+2is}\right\|_2
+ \left\|\varrho_2^{is}(\varrho_2^{1/2}-\varrho_1^{1/2})\right\|_2 \\
&\leq \sqrt{1+4s^2} \left\|\varrho_1^{1/2}- \varrho_2^{1/2}\right\|_2
+ \left\|\varrho_1^{1/2}-\varrho_2^{1/2}\right\|_2, \label{eq:cgLipschitzestimate}
\end{align}
where in the second line we used $\|\varrho_2^{is}\|_\infty\leq1$ and the Schatten-$2$ Lipschitz functional-calculus bound \cite{BirmanSolomyak2003}
\begin{equation}
\left\|f(\varrho_1^{1/2})-f(\varrho_2^{1/2})\right\|_2
\leq\operatorname{Lip}(f)\left\|\varrho_1^{1/2}-\varrho_2^{1/2}\right\|_2
\end{equation}
for $f(x):=x^{1+2is}$. Since $|f'(x)|=\sqrt{1+4s^2}$ for $x>0$, one has $\operatorname{Lip}(f)=\sqrt{1+4s^2}$.  Finally, the Powers--St{\o}rmer inequality \cite{PowersStormer1970}, applied to $\varrho_1$ and $\varrho_2$, gives
\begin{equation}
\lt\|\varrho_1^{1/2}-\varrho_2^{1/2}\rt\|_2^2
\leq\|\varrho_1-\varrho_2\|_1=\Delta_\rho.
\end{equation}
This together with \eqref{eq:cgLipschitzestimate} proves \eqref{eq:cgoperatorcontinuity}.

We now use \eqref{eq:cgoperatorcontinuity} to find
\begin{align}
\Bigl\|(\varrho_1^{is}-\varrho_2^{is})|\phi^1\rangle\Bigr\|^2
&= \tr_r\!\left[(\varrho_1^{is}-\varrho_2^{is})^\dagger
(\varrho_1^{is}-\varrho_2^{is})\, \t_r^1 \right]\\
&\leq \left\|(\varrho_1^{is}-\varrho_2^{is})^\dagger
(\varrho_1^{is}-\varrho_2^{is}) \varrho_1^{1/2}\right\|_2
\left\|\varrho_1^{-1/2}\t_r^1\right\|_2\\
&\leq 2\epsilon_{tail}^{-1/2}
\left\|(\varrho_1^{is}-\varrho_2^{is})\varrho_1^{1/2}\right\|_2
\leq2C_s\sqrt{\frac{\Delta_\rho}{\epsilon_{tail}}},
\end{align}
where in the final line we used $\|(\varrho_1^{is}-\varrho_2^{is})^\dagger\|_\infty \leq 2$ as well as $\|\varrho_1^{-1/2}\t_r^1\|_2^2\leq\epsilon_{tail}^{-1}$ from relative log-stability \eqref{eq:relativelogstabilitybulk} (see Lemma~8 of Ref.~\cite{JLMS}).

Taking the square root gives the middle term in \eqref{eq:cgflowreplace}.  The triangle inequality adds $\Delta_\phi$ when $|\phi^1\rangle$ is replaced by $|\phi^2\rangle$, while the difference of the $A_{cg}^\a$ phases is at most $|s|\Delta_A/(4G)$ and causes an error at most of this amount according to \er{eq:unitaryflowdifference}.  This proves \eqref{eq:cgflowreplace}.
\end{proof}

\addcontentsline{toc}{section}{References}
\bibliographystyle{JHEP}
\bibliography{references}

@article{Faulkner:2022ada,
    author = "Faulkner, Thomas and Li, Min",
    title = "{Asymptotically isometric codes for holography}",
    eprint = "2211.12439",
    archivePrefix = "arXiv",
    primaryClass = "hep-th",
    month = "11",
    year = "2022"
}

@article{Crann:2024gkv,
    author = "Crann, Jason and Kang, Monica Jinwoo",
    title = "{Algebraic approach to spacetime bulk reconstruction}",
    eprint = "2412.00298",
    archivePrefix = "arXiv",
    primaryClass = "math.OA",
    month = "11",
    year = "2024"
}

@article{Kuchar:1994zk,
    author = "Kuchar, Karel V.",
    title = "{Geometrodynamics of Schwarzschild black holes}",
    eprint = "gr-qc/9403003",
    archivePrefix = "arXiv",
    reportNumber = "UU-REL-94-3-1",
    doi = "10.1103/PhysRevD.50.3961",
    journal = "Phys. Rev. D",
    volume = "50",
    pages = "3961--3981",
    year = "1994"
}

@article{Kaplan:2022orm,
    author = "Kaplan, Molly and Marolf, Donald",
    title = "{The action of HRT-areas as operators in semiclassical gravity}",
    eprint = "2203.04270",
    archivePrefix = "arXiv",
    primaryClass = "hep-th",
    doi = "10.1007/JHEP08(2022)102",
    journal = "JHEP",
    volume = "08",
    pages = "102",
    year = "2022"
}

@article{Harlow:2018tqv,
    author = "Harlow, Daniel and Jafferis, Daniel",
    title = "{The Factorization Problem in Jackiw-Teitelboim Gravity}",
    eprint = "1804.01081",
    archivePrefix = "arXiv",
    primaryClass = "hep-th",
    doi = "10.1007/JHEP02(2020)177",
    journal = "JHEP",
    volume = "02",
    pages = "177",
    year = "2020"
}

@article{Maldacena:2001kr,
    author = "Maldacena, Juan Martin",
    title = "{Eternal black holes in anti-de Sitter}",
    eprint = "hep-th/0106112",
    archivePrefix = "arXiv",
    reportNumber = "NSF-ITP-01-59",
    doi = "10.1088/1126-6708/2003/04/021",
    journal = "JHEP",
    volume = "04",
    pages = "021",
    year = "2003"
}

@article{Thiemann:1992jj,
    author = "Thiemann, T. and Kastrup, H. A.",
    title = "{Canonical quantization of spherically symmetric gravity in Ashtekar's selfdual representation}",
    eprint = "gr-qc/9310012",
    archivePrefix = "arXiv",
    reportNumber = "PITHA-92-23",
    doi = "10.1016/0550-3213(93)90623-W",
    journal = "Nucl. Phys. B",
    volume = "399",
    pages = "211--258",
    year = "1993"
}

@article{Kastrup:1993br,
    author = "Kastrup, H. A. and Thiemann, T.",
    title = "{Spherically symmetric gravity as a completely integrable system}",
    eprint = "gr-qc/9401032",
    archivePrefix = "arXiv",
    reportNumber = "PITHA-93-35",
    doi = "10.1016/0550-3213(94)90293-3",
    journal = "Nucl. Phys. B",
    volume = "425",
    pages = "665--686",
    year = "1994"
}

@article{Almheiri:2019hni,
    author = "Almheiri, Ahmed and Mahajan, Raghu and Maldacena, Juan and Zhao, Ying",
    title = "{The Page curve of Hawking radiation from semiclassical geometry}",
    eprint = "1908.10996",
    archivePrefix = "arXiv",
    primaryClass = "hep-th",
    doi = "10.1007/JHEP03(2020)149",
    journal = "JHEP",
    volume = "03",
    pages = "149",
    year = "2020"
}

@article{Kim:2021gjx,
    author = "Kim, Isaac H. and Shi, Bowen and Kato, Kohtaro and Albert, Victor V.",
    title = "{Chiral Central Charge from a Single Bulk Wave Function}",
    eprint = "2110.06932",
    archivePrefix = "arXiv",
    primaryClass = "quant-ph",
    doi = "10.1103/PhysRevLett.128.176402",
    journal = "Phys. Rev. Lett.",
    volume = "128",
    number = "17",
    pages = "176402",
    year = "2022"
}

@article{Marolf:2020vsi,
    author = "Marolf, Donald and Wang, Shannon and Wang, Zhencheng",
    title = "{Probing phase transitions of holographic entanglement entropy with fixed area states}",
    eprint = "2006.10089",
    archivePrefix = "arXiv",
    primaryClass = "hep-th",
    doi = "10.1007/JHEP12(2020)084",
    journal = "JHEP",
    volume = "12",
    pages = "084",
    year = "2020"
}

@article{Akers:2020pmf,
    author = "Akers, Chris and Penington, Geoff",
    title = "{Leading order corrections to the quantum extremal surface prescription}",
    eprint = "2008.03319",
    archivePrefix = "arXiv",
    primaryClass = "hep-th",
    doi = "10.1007/JHEP04(2021)062",
    journal = "JHEP",
    volume = "04",
    pages = "062",
    year = "2021"
}

@article{Jafferis:2014lza,
    author = "Jafferis, Daniel L. and Suh, S. Josephine",
    title = "{The Gravity Duals of Modular Hamiltonians}",
    eprint = "1412.8465",
    archivePrefix = "arXiv",
    primaryClass = "hep-th",
    reportNumber = "MIT-CTP-4611",
    doi = "10.1007/JHEP09(2016)068",
    journal = "JHEP",
    volume = "09",
    pages = "068",
    year = "2016"
}

@article{Bousso:2019dxk,
    author = "Bousso, Raphael and Chandrasekaran, Venkatesa and Shahbazi-Moghaddam, Arvin",
    title = "{From black hole entropy to energy-minimizing states in QFT}",
    eprint = "1906.05299",
    archivePrefix = "arXiv",
    primaryClass = "hep-th",
    doi = "10.1103/PhysRevD.101.046001",
    journal = "Phys. Rev. D",
    volume = "101",
    number = "4",
    pages = "046001",
    year = "2020"
}

@article{Ryu:2006ef,
    author = "Ryu, Shinsei and Takayanagi, Tadashi",
    title = "{Aspects of Holographic Entanglement Entropy}",
    eprint = "hep-th/0605073",
    archivePrefix = "arXiv",
    reportNumber = "NSF-KITP-06-31, KUNS-2021",
    doi = "10.1088/1126-6708/2006/08/045",
    journal = "JHEP",
    volume = "08",
    pages = "045",
    year = "2006"
}

@article{Hubeny:2007xt,
    author = "Hubeny, Veronika E. and Rangamani, Mukund and Takayanagi, Tadashi",
    title = "{A Covariant holographic entanglement entropy proposal}",
    eprint = "0705.0016",
    archivePrefix = "arXiv",
    primaryClass = "hep-th",
    reportNumber = "DCPT-07-13, KUNS-2069",
    doi = "10.1088/1126-6708/2007/07/062",
    journal = "JHEP",
    volume = "07",
    pages = "062",
    year = "2007"
}

@article{Penington:2019npb,
    author = "Penington, Geoffrey",
    title = "{Entanglement Wedge Reconstruction and the Information Paradox}",
    eprint = "1905.08255",
    archivePrefix = "arXiv",
    primaryClass = "hep-th",
    doi = "10.1007/JHEP09(2020)002",
    journal = "JHEP",
    volume = "09",
    pages = "002",
    year = "2020"
}

@article{Chen:2019gbt,
    author = "Chen, Chi-Fang and Penington, Geoffrey and Salton, Grant",
    title = "{Entanglement Wedge Reconstruction using the Petz Map}",
    eprint = "1902.02844",
    archivePrefix = "arXiv",
    primaryClass = "hep-th",
    doi = "10.1007/JHEP01(2020)168",
    journal = "JHEP",
    volume = "01",
    pages = "168",
    year = "2020"
}

@article{Dong:2017xht,
    author = "Dong, Xi and Lewkowycz, Aitor",
    title = "{Entropy, Extremality, Euclidean Variations, and the Equations of Motion}",
    eprint = "1705.08453",
    archivePrefix = "arXiv",
    primaryClass = "hep-th",
    doi = "10.1007/JHEP01(2018)081",
    journal = "JHEP",
    volume = "01",
    pages = "081",
    year = "2018"
}

@article{Dong:2018seb,
    author = "Dong, Xi and Harlow, Daniel and Marolf, Donald",
    title = "{Flat entanglement spectra in fixed-area states of quantum gravity}",
    eprint = "1811.05382",
    archivePrefix = "arXiv",
    primaryClass = "hep-th",
    doi = "10.1007/JHEP10(2019)240",
    journal = "JHEP",
    volume = "10",
    pages = "240",
    year = "2019"
}

@article{Akers:2018fow,
    author = "Akers, Chris and Rath, Pratik",
    title = "{Holographic Renyi Entropy from Quantum Error Correction}",
    eprint = "1811.05171",
    archivePrefix = "arXiv",
    primaryClass = "hep-th",
    doi = "10.1007/JHEP05(2019)052",
    journal = "JHEP",
    volume = "05",
    pages = "052",
    year = "2019"
}

@article{Jafferis:2015del,
    author = "Jafferis, Daniel L. and Lewkowycz, Aitor and Maldacena, Juan and Suh, S. Josephine",
    title = "{Relative entropy equals bulk relative entropy}",
    eprint = "1512.06431",
    archivePrefix = "arXiv",
    primaryClass = "hep-th",
    reportNumber = "NSF-KITP-15-162",
    doi = "10.1007/JHEP06(2016)004",
    journal = "JHEP",
    volume = "06",
    pages = "004",
    year = "2016"
}

@article{Lewkowycz:2013nqa,
    author = "Lewkowycz, Aitor and Maldacena, Juan",
    title = "{Generalized gravitational entropy}",
    eprint = "1304.4926",
    archivePrefix = "arXiv",
    primaryClass = "hep-th",
    doi = "10.1007/JHEP08(2013)090",
    journal = "JHEP",
    volume = "08",
    pages = "090",
    year = "2013"
}

@article{Almheiri:2014lwa,
    author = "Almheiri, Ahmed and Dong, Xi and Harlow, Daniel",
    title = "{Bulk Locality and Quantum Error Correction in AdS/CFT}",
    eprint = "1411.7041",
    archivePrefix = "arXiv",
    primaryClass = "hep-th",
    reportNumber = "SU-ITP-14-30",
    doi = "10.1007/JHEP04(2015)163",
    journal = "JHEP",
    volume = "04",
    pages = "163",
    year = "2015"
}

@article{Chen:2018rgz,
    author = "Chen, Yiming and Dong, Xi and Lewkowycz, Aitor and Qi, Xiao-Liang",
    title = "{Modular Flow as a Disentangler}",
    eprint = "1806.09622",
    archivePrefix = "arXiv",
    primaryClass = "hep-th",
    doi = "10.1007/JHEP12(2018)083",
    journal = "JHEP",
    volume = "12",
    pages = "083",
    year = "2018"
}

@article{Dong:2019piw,
    author = "Dong, Xi and Marolf, Donald",
    title = "{One-loop universality of holographic codes}",
    eprint = "1910.06329",
    archivePrefix = "arXiv",
    primaryClass = "hep-th",
    doi = "10.1007/JHEP03(2020)191",
    journal = "JHEP",
    volume = "03",
    pages = "191",
    year = "2020"
}

@article{Dong:2013qoa,
    author = "Dong, Xi",
    title = "{Holographic Entanglement Entropy for General Higher Derivative Gravity}",
    eprint = "1310.5713",
    archivePrefix = "arXiv",
    primaryClass = "hep-th",
    reportNumber = "SU-ITP-13-21",
    doi = "10.1007/JHEP01(2014)044",
    journal = "JHEP",
    volume = "01",
    pages = "044",
    year = "2014"
}

@article{Zou:2022nuj,
    author = "Zou, Yijian and Shi, Bowen and Sorce, Jonathan and Lim, Ian T. and Kim, Isaac H.",
    title = "{Modular Commutators in Conformal Field Theory}",
    eprint = "2206.00027",
    archivePrefix = "arXiv",
    primaryClass = "cond-mat.str-el",
    doi = "10.1103/PhysRevLett.129.260402",
    journal = "Phys. Rev. Lett.",
    volume = "129",
    number = "26",
    pages = "260402",
    year = "2022"
}

@article{Kim:2021tse,
    author = "Kim, Isaac H. and Shi, Bowen and Kato, Kohtaro and Albert, Victor V.",
    title = "{Modular commutator in gapped quantum many-body systems}",
    eprint = "2110.10400",
    archivePrefix = "arXiv",
    primaryClass = "quant-ph",
    doi = "10.1103/PhysRevB.106.075147",
    journal = "Phys. Rev. B",
    volume = "106",
    number = "7",
    pages = "075147",
    year = "2022"
}

@article{Kudler-Flam:2022jwd,
    author = "Kudler-Flam, Jonah and Rath, Pratik",
    title = "{Large and small corrections to the JLMS Formula from replica wormholes}",
    eprint = "2203.11954",
    archivePrefix = "arXiv",
    primaryClass = "hep-th",
    doi = "10.1007/JHEP08(2022)189",
    journal = "JHEP",
    volume = "08",
    pages = "189",
    year = "2022"
}

@article{Dong:2025orj,
    author = "Dong, Xi and Marolf, Donald and Rath, Pratik",
    title = "{Geometric entropies and their Hamiltonian flows}",
    eprint = "2501.12438",
    archivePrefix = "arXiv",
    primaryClass = "hep-th",
    doi = "10.1007/JHEP05(2025)085",
    journal = "JHEP",
    volume = "05",
    pages = "085",
    year = "2025"
}

@article{JLMS,
   author = "Dong, Xi and Marolf, Donald and Rath, Pratik",
    title = "{The JLMS formula in a large code with approximate error correction}",
    eprint = "2601.00442",
    archivePrefix = "arXiv",
    primaryClass = "hep-th",
    month = "1",
    year = "2026"
}

@article{Kaplan:2023oaj,
    author = "Kaplan, Molly",
    title = "{The action of geometric entropy in topologically massive gravity}",
    eprint = "2308.09763",
    archivePrefix = "arXiv",
    primaryClass = "hep-th",
    doi = "10.1007/JHEP12(2023)106",
    journal = "JHEP",
    volume = "12",
    pages = "106",
    year = "2023"
}

@article{Susskind:1995qc,
    author = "Susskind, Leonard and Uglum, John",
    title = "{String physics and black holes}",
    eprint = "hep-th/9511227",
    archivePrefix = "arXiv",
    reportNumber = "SU-ITP-95-31",
    doi = "10.1016/0920-5632(95)00630-3",
    journal = "Nucl. Phys. B Proc. Suppl.",
    volume = "45BC",
    pages = "115--134",
    year = "1996"
}

@article{Carlip:1993sa,
    author = "Carlip, Steven and Teitelboim, Claudio",
    title = "{The Off-shell black hole}",
    eprint = "gr-qc/9312002",
    archivePrefix = "arXiv",
    reportNumber = "IASSNS-HEP-93-84, UCD-93-34",
    doi = "10.1088/0264-9381/12/7/011",
    journal = "Class. Quant. Grav.",
    volume = "12",
    pages = "1699--1704",
    year = "1995"
}

@article{BirmanSolomyak2003,
  author  = {Birman, Mikhail Sh. and Solomyak, Michael Z.},
  title   = {Double Operator Integrals in a Hilbert Space},
  journal = {Integral Equations and Operator Theory},
  volume  = {47},
  number  = {2},
  pages   = {131--168},
  year    = {2003},
  doi     = {10.1007/s00020-003-1157-8}
}

@article{PowersStormer1970,
  author  = {Powers, Robert T. and St{\o}rmer, Erling},
  title   = {Free States of the Canonical Anticommutation Relations},
  journal = {Communications in Mathematical Physics},
  volume  = {16},
  number  = {1},
  pages   = {1--33},
  year    = {1970},
  doi     = {10.1007/BF01645492}
}

@article{RG,
    author = "Dong, Xi and Marolf, Donald and Rath, Pratik",
    title = "{Holographic codes and bulk RG flows}",
    eprint = "2509.21438",
    archivePrefix = "arXiv",
    primaryClass = "hep-th",
    month = "9",
    year = "2025"
}

@article{Ceyhan:2018zfg,
    author = "Ceyhan, Fikret and Faulkner, Thomas",
    title = "{Recovering the QNEC from the ANEC}",
    eprint = "1812.04683",
    archivePrefix = "arXiv",
    primaryClass = "hep-th",
    doi = "10.1007/s00220-020-03751-y",
    journal = "Commun. Math. Phys.",
    volume = "377",
    number = "2",
    pages = "999--1045",
    year = "2020"
}

@article{Chua:2023srl,
    author = "Chua, Wan Zhen and Hartman, Thomas",
    title = "{Black hole wavefunctions and microcanonical states}",
    eprint = "2309.05041",
    archivePrefix = "arXiv",
    primaryClass = "hep-th",
    doi = "10.1007/JHEP06(2024)054",
    journal = "JHEP",
    volume = "06",
    pages = "054",
    year = "2024"
}

@article{Lewkowycz:2018sgn,
    author = "Lewkowycz, Aitor and Parrikar, Onkar",
    title = "{The holographic shape of entanglement and Einstein\textquoteright{}s equations}",
    eprint = "1802.10103",
    archivePrefix = "arXiv",
    primaryClass = "hep-th",
    doi = "10.1007/JHEP05(2018)147",
    journal = "JHEP",
    volume = "05",
    pages = "147",
    year = "2018"
}

@article{Dong:2016eik,
    author = "Dong, Xi and Harlow, Daniel and Wall, Aron C.",
    title = "{Reconstruction of Bulk Operators within the Entanglement Wedge in Gauge-Gravity Duality}",
    eprint = "1601.05416",
    archivePrefix = "arXiv",
    primaryClass = "hep-th",
    reportNumber = "NSF-KITP-16-005",
    doi = "10.1103/PhysRevLett.117.021601",
    journal = "Phys. Rev. Lett.",
    volume = "117",
    number = "2",
    pages = "021601",
    year = "2016"
}

@article{Ryu:2006bv,
    author = "Ryu, Shinsei and Takayanagi, Tadashi",
    title = "{Holographic derivation of entanglement entropy from AdS/CFT}",
    eprint = "hep-th/0603001",
    archivePrefix = "arXiv",
    reportNumber = "NSF-KITP-06-11",
    doi = "10.1103/PhysRevLett.96.181602",
    journal = "Phys. Rev. Lett.",
    volume = "96",
    pages = "181602",
    year = "2006"
}

@article{Cotler:2017erl,
    author = "Cotler, Jordan and Hayden, Patrick and Penington, Geoffrey and Salton, Grant and Swingle, Brian and Walter, Michael",
    title = "{Entanglement Wedge Reconstruction via Universal Recovery Channels}",
    eprint = "1704.05839",
    archivePrefix = "arXiv",
    primaryClass = "hep-th",
    doi = "10.1103/PhysRevX.9.031011",
    journal = "Phys. Rev. X",
    volume = "9",
    number = "3",
    pages = "031011",
    year = "2019"
}

@article{Faulkner:2017vdd,
    author = "Faulkner, Thomas and Lewkowycz, Aitor",
    title = "{Bulk locality from modular flow}",
    eprint = "1704.05464",
    archivePrefix = "arXiv",
    primaryClass = "hep-th",
    doi = "10.1007/JHEP07(2017)151",
    journal = "JHEP",
    volume = "07",
    pages = "151",
    year = "2017"
}

@article{Bousso:2020yxi,
    author = "Bousso, Raphael and Chandrasekaran, Venkatesa and Rath, Pratik and Shahbazi-Moghaddam, Arvin",
    title = "{Gravity dual of Connes cocycle flow}",
    eprint = "2007.00230",
    archivePrefix = "arXiv",
    primaryClass = "hep-th",
    doi = "10.1103/PhysRevD.102.066008",
    journal = "Phys. Rev. D",
    volume = "102",
    number = "6",
    pages = "066008",
    year = "2020"
}

@article{Faulkner:2018faa,
    author = "Faulkner, Thomas and Li, Min and Wang, Huajia",
    title = "{A modular toolkit for bulk reconstruction}",
    eprint = "1806.10560",
    archivePrefix = "arXiv",
    primaryClass = "hep-th",
    doi = "10.1007/JHEP04(2019)119",
    journal = "JHEP",
    volume = "04",
    pages = "119",
    year = "2019"
}

@article{Faulkner:2013ana,
    author = "Faulkner, Thomas and Lewkowycz, Aitor and Maldacena, Juan",
    title = "{Quantum corrections to holographic entanglement entropy}",
    eprint = "1307.2892",
    archivePrefix = "arXiv",
    primaryClass = "hep-th",
    doi = "10.1007/JHEP11(2013)074",
    journal = "JHEP",
    volume = "11",
    pages = "074",
    year = "2013"
}

@article{Engelhardt:2014gca,
    author = "Engelhardt, Netta and Wall, Aron C.",
    title = "{Quantum Extremal Surfaces: Holographic Entanglement Entropy beyond the Classical Regime}",
    eprint = "1408.3203",
    archivePrefix = "arXiv",
    primaryClass = "hep-th",
    doi = "10.1007/JHEP01(2015)073",
    journal = "JHEP",
    volume = "01",
    pages = "073",
    year = "2015"
}

@article{Harlow:2016vwg,
    author = "Harlow, Daniel",
    title = "{The Ryu\textendash{}Takayanagi Formula from Quantum Error Correction}",
    eprint = "1607.03901",
    archivePrefix = "arXiv",
    primaryClass = "hep-th",
    doi = "10.1007/s00220-017-2904-z",
    journal = "Commun. Math. Phys.",
    volume = "354",
    number = "3",
    pages = "865--912",
    year = "2017"
}

%

% %~~~~~~~~~~~~~~~~~~~~~~~~~~~~~~~~~~~~~~~~~~~~~~~~~~~~~~~~~~~~~~~~~~~~~
\end{document}